\documentclass[sigconf, nonacm]{acmart}

\usepackage{amssymb}
\def\ojoin{\setbox0=\hbox{$\bowtie$}%
  \rule[-.02ex]{.25em}{.4pt}\llap{\rule[\ht0]{.25em}{.4pt}}}
\def\leftouterjoin{\mathbin{\ojoin\mkern-5.8mu\bowtie}}
\def\rightouterjoin{\mathbin{\bowtie\mkern-5.8mu\ojoin}}
\def\fullouterjoin{\mathbin{\ojoin\mkern-5.8mu\bowtie\mkern-5.8mu\ojoin}}

\usepackage{graphicx}
\usepackage{balance}  %
\usepackage{float}
\usepackage{multirow}
\usepackage{caption}
\usepackage{enumitem}
\usepackage[linesnumbered,ruled,vlined,noend]{algorithm2e}
\SetKwInput{KwInput}{Input}                %
\SetKwInput{KwOutput}{Output}              %
\usepackage{wrapfig}

\usepackage{hyperref}

\usepackage{xcolor}
\usepackage{bm}

\usepackage{tikz}
\usetikzlibrary{shapes,arrows}

\newcommand{\jedi}{\textsf{JEDI~}}

\newcommand{\jfd}{\textsf{Join FD discovery~}}
\newcommand{\fdSet}[0]{\ensuremath{\mathcal{D}}}

\newcommand{\fun}[1]{\ensuremath{\mathrm{#1}}\xspace}
\newcommand{\compUpstaged}{\texttt{upstagedFDs}\xspace}
\newcommand{\subInfer}{\texttt{infer}\xspace}
\newcommand{\subRefine}{\texttt{refine}\xspace}

\newcommand{\add}{{\bf add}\xspace}
\newcommand{\prune}{{\bf prune}\xspace}

\newcommand{\leftInst}[0]{\ensuremath{L}\xspace}
\newcommand{\rightInst}[0]{\ensuremath{R}\xspace}
\newcommand{\leftSch}[0]{\ensuremath{\mathbf{S}}\xspace}
\newcommand{\rightSch}[0]{\ensuremath{\mathbf{T}}\xspace}
\newcommand{\leftJoinAtt}[0]{\ensuremath{X}\xspace}
\newcommand{\rightJoinAtt}[0]{\ensuremath{Y}\xspace}

\newtheorem{definition}{Definition}
\newtheorem{theorem}{Theorem}

\newtheorem{example}{Example}

\newcommand {\approxi}[1]{\rightharpoondown^{#1}}
\newcommand {\approxfd} {\approxi{\varepsilon}}

\begin{document}

\title{Discovering Multi-Table Functional Dependencies Without Full Join Computation}

\author{Ugo Comignani}
\affiliation{%
  \institution{Tyrex team, Univ. Grenoble Alpes, CNRS, Inria,Grenoble INP, LIG}
  \city{Grenoble}
  \country{France}
}
\email{ugo.comignani@inria.fr}

\author{Laure Berti-\'Equille}
\affiliation{%
  \institution{IRD, UMR ESPACE DEV, Univ de Toulon, AMU, CNRS, LIS, DIAMS}
  \city{Montpellier}
  \country{France}
}
\email{laure.berti@ird.fr}

\author{Noël Novelli}
\affiliation{%
  \institution{Aix Marseille Univ, Université de Toulon, CNRS, LIS, DIAMS}
  \city{Marseille}
  \country{France}
}
\email{noel.novelli@lis-lab.fr}

\setlength{\abovedisplayskip}{3pt}
\setlength{\belowdisplayskip}{3pt}

\begin{abstract}
In this paper, we study the problem of discovering \it{join FDs}, i.e., functional dependencies (FDs) that hold on multiple joined tables. We leverage logical inference, selective mining, and sampling and show that we can discover most of the exact join FDs from the single tables participating to the join and avoid the full computation of the join result. We propose algorithms to speed-up the join FD discovery process and mine FDs on the fly only from necessary data partitions. We introduce \jedi ({\it Join dEpendency DIscovery}), our solution to discover join FDs without computation of the full join beforehand. Our experiments on a range of real-world and synthetic data demonstrate the benefits of our method over existing FD discovery methods that need to precompute the join results before discovering the FDs. We show that the performance depends on the cardinalities and coverage of the join attribute values: for join operations with low coverage, \jedi  with selective mining outperforms the competing methods using the straightforward approach of full join computation by one order of magnitude in terms of runtime and can discover three-quarters of the exact join FDs using mainly logical inference in half of its total execution time on average. For higher join coverage, \jedi  with sampling reaches precision of 1 with only 63\% of the table input size on average.

\end{abstract}

\maketitle

\section{Introduction}

The problem of computing all functional dependencies (FDs) that hold on a given relational table has received much attention from the academia and the industry due to numerous application demands. FDs capture deterministic relations between attributes of a database. Typically, an FD $X \rightarrow Y$ with column sets $X$ and $Y$ in a given table expresses that the combination of values in $X$ columns uniquely determines the value of every column in $Y$. This can be of critical use for the database design process, table decomposition, database normalization, and for many data management applications, such as data cleaning \cite{Thirumuruganathan17}, data profiling \cite{ChuIP13}, and query optimization \cite{IlyasMHBA04, Paulley}.  Although FDs are essential for many aspects of database management, the problem of discovering FDs in an entire database with multiple joined tables has been considered so far only by computing FDs from single tables in isolation without much investigation on the effect of the join operation on FD discovery. Trivially, one can compute the join results between multiple tables beforehand and then run an existing method to discover FDs. However, such computation could be significantly reduced by inferring FDs from each separate table and computing the join result only when and with what is needed. In this paper, we address the problem of frugal computation of join FDs and rename it as the \textsf{Join FD discovery~} problem that aims at finding  FDs that hold on the join result of multiple relational tables without computing the full join operations.  

The \jfd problem is interesting for many reasons: (1) Join FDs can be used to assess whether the FDs discovered from a single table are meaningful for  the entire database when they still hold in the join results of multiple tables. Inversely, join FDs can be used to prune the set of FDs discovered from a single table when they are not persistent across the joined tables. In certain cases, they may hold incidentally on a single table due to some data errors or table incompleteness; (2) New FDs can be discovered through join operations and they can be useful for database administration, query optimization \cite{Paulley}, and other application purposes. Therefore, it is beneficial to  discover them efficiently; and  (3) Most join FDs can be inferred from the single tables participating to the join. This can save computation time since there is no need to compute the join result entirely before FD discovery. To address the \jfd problem, we combine three strategies in one approach: logical inference, selective mining (with vertical data partitions), and selective sampling (with horizontal partitions) for FD discovery and we propose {\textsf{JEDI}}, an efficient solution for automatically discovering multi-table join FDs with frugal computation, rather than the trivial and costly approach of computing the result of a join operation tables before FD discovery. 

 \begin{figure*}[t]
	\centering
	\includegraphics[width=.9\linewidth]{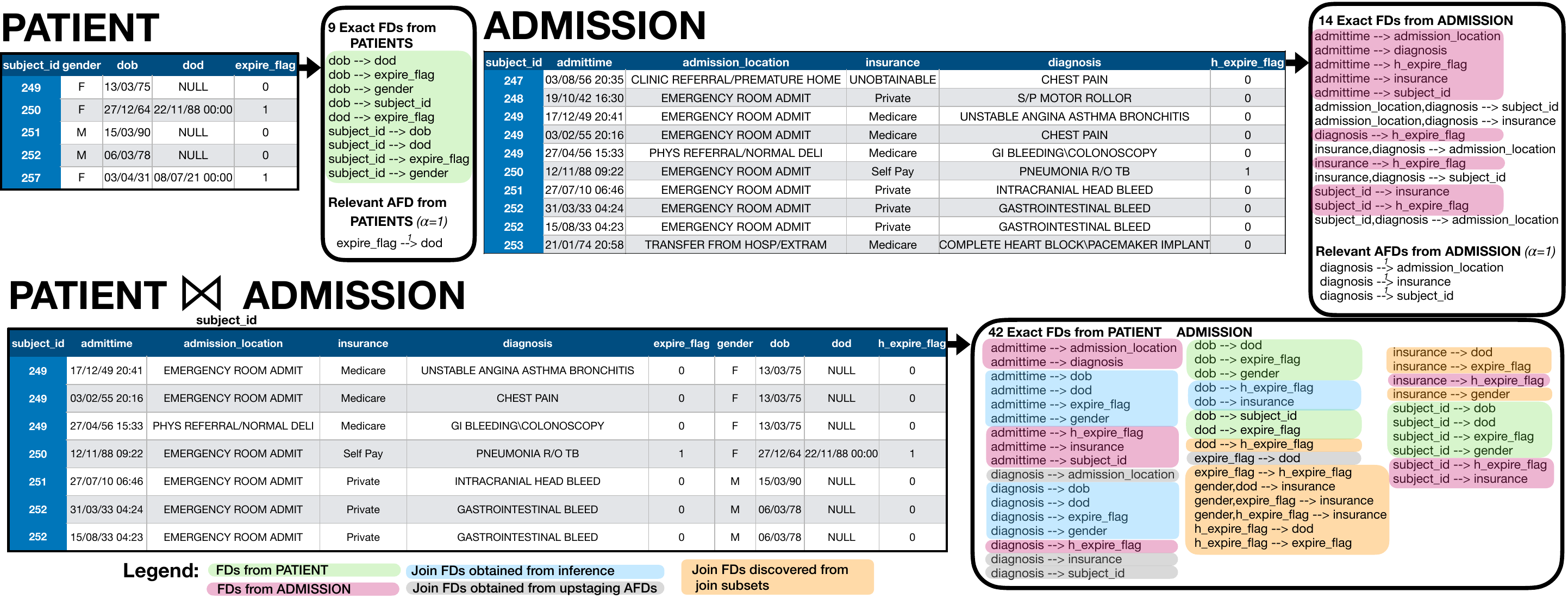}
	\caption{Illustrative Example}\label{fig:example}
\end{figure*} 

{\bf Challenges.} The first challenge is the exponential complexity of join FD discovery. The straightforward approach is very costly for two reasons: (1) the numbers of attributes of the tables participating in the join are added. Many of the existing methods exhibit poor scalability as the number of attributes increases (despite the efficiency of their pruning methods to search over the lattice of attribute combinations).  This problem is already well-known for FD discovery from a single table  \cite{Kruse18} but it is amplified in the join FD discovery problem; and (2) the execution times of the join and the FD discovery are also cumulated.

As a second challenge, real-world data sets are plagued with missing or erroneous values that may lead to spurious FDs or low recall in terms of  genuine FDs \cite{Berti-EquilleHN18}. More particularly, the problem of fully missing records has not been much investigated although it may have a significant impact on the validity of the FD discovery results. We can also observe disparate value cardinalities in the join attributes of the participating tables. The assumption of  {\it Preservation of Join Value Sets} \cite{booksGarcia}  is often violated due to many “dangling tuples”, that is, for $R\bowtie S$,  tuples of $R$ that join with no tuple of $S$ (and inversely). Moreover, some join attribute values  (e.g., keys) may be repeated and joined multiple times. Because of these disparities and partial coverage, the union of the FD sets obtained from single tables is not equivalent to the FD set obtained from the join result. %

Finally, pruning the join FD search space efficiently is challenging. Existing methods identify minimal FDs from a single table. However,    minimality does not guarantee that the set of discovered FDs will be parsimonious \cite{Kruse18,Zhang2020} and minimal FDs from single tables are not necessarily minimal multi-table FDs. In this paper, we propose the first approach that discover efficiently join FDs from multiple tables of a database. %

The main contributions of our paper are the following:
\begin{itemize}[noitemsep, nolistsep, leftmargin=*,topsep=0pt]
    \item We demonstrate several useful properties for inferring exact FDs that can be obtained from various join operators between two tables (i.e., inner join,  left and right semi-joins, left and right outer joins) without computing the full join result beforehand;
    \item We propose four algorithms that leverage these properties, and efficiently compute the exact FDs from two or more tables without  computing full joins;
    \item We propose {\textsf{JEDI}}, a system implementing our algorithms. \jedi is  available at \url{https://github.com/xxxx} with code, scripts, and data sets for the reproducibility of our experiments;
    \item We provide the results of an extensive experimental evaluation: we compare \jedi against three state-of-the-art FD discovery methods over a diverse array of real-world and synthetic data sets with varying numbers of tuples, attributes, domain sizes, types and coverage of the join operations. We find that \jedi outperforms the competing methods by one order of magnitude in terms of execution time for discovering exact FDs while preserving the smallest memory consumption on average.
\end{itemize}

\noindent {\bf Outline.} Section~\ref{sec:example} presents an illustrative example. Section presents the necessary background and notations. In Section~\ref{sec:preliminaries}, we formalize the \textsf{Join FD discovery~} problem and provide an overview of {\textsf{JEDI}}. In Section~\ref{sec:contributions}, we present our main contributions as the algorithms at the core of {\textsf{JEDI}}. We describe our performance experiments evaluating the efficiency and precision of {\textsf{JEDI}} in  Section~\ref{sec:experiments}. Finally, we discuss related work in Section~\ref{sec:relatedwork} and conclude in Section~\ref{sec:conclusion}.

\section{An Illustrative Example}
\label{sec:example}
To motivate our approach, we consider the following example illustrated in Figure 1. Samples of two tables are extracted from the clinical database \texttt{ MIMIC-III}\footnote{\url{https://physionet.org/content/mimiciii/1.4/}} \cite{mimiciii}:  {\small\texttt{PATIENT}} table contains information of 5 patients: their identifier ({\small\texttt{subject\_id}}), \texttt{gender}, date of birth ({\small\texttt{dob}}), date of death ({\small\texttt{dod}}), and \mbox{{\small\texttt{expire\_flag}}}.  {\small\texttt{ADMISSION}} table contains some administrative and clinical  information about each patient such as the hospital admission time ({\small\texttt{admittime}}), the admission location ({\small\texttt{admission\_location}}), the {\small\texttt{insurance}}, the {\small\texttt{diagnosis}}, and  {\small\texttt{h\_expire\_flag}} indicating whether the patient died at the hospital. The table contains 10 admission records. We can opt for the NULL semantics where missing values are all equal and extract exact FDs from each table. We obtain  9 exact FDs for   {\small\texttt{PATIENT}}  and 14 exact FDs for  {\small\texttt{ADMISSION}}.  The figure also shows relevant approximate FDs (with degree $\alpha=1$) from both tables. In addition, the figure presents the join result of the two tables over the join attribute {\small\texttt{subject\_id}}. %
In our example, %
 42 exact minimal and canonical FDs can be discovered from the join result. Interestingly, we can observe that the 9 exact FDs from the  {\small\texttt{PATIENT}} table in the left side of the join are preserved (in green), as well as 9 (out of the 14) exact FDs from the {\small\texttt{ADMISSION}} table (in pink). Additionally, 10 FDs discovered from the join result can be obtained by inference between the sets of exact FDs discovered respectively from each single table (in blue). For example,
${\small\texttt{admittime}}\to {\small\texttt{dob}}$ is obtained from ${\small\texttt{admittime}}\to  {\small\texttt{subject\_id}}$ in {\small\texttt{ADMISSION}} and  ${\small\texttt{subject\_id}} \to  {\small\texttt{dob}}$ in {\small\texttt{PATIENT}}. In the {\small\texttt{PATIENT}} table, the ${\small\texttt{subject\_id}}$ values are unique, whereas in the {\small\texttt{ADMISSION}} table, we can observe two repeated records for patients \#252 and \#249. Moreover, the patient \#257 is absent in the {\small\texttt{ADMISSION}} table, whereas s/he is present in the table \texttt{PATIENT}, similarly patient \#247 is absent in the {\small\texttt{PATIENT}} table but present {\small\texttt{ADMISSION}} table. Due to these cardinality disparities in the domain of the join attribute  {\small\texttt{subject\_id}}, some approximate FDs discovered from the single tables can become exact FDs in the join result  (we will use the term ``upstaged'' AFDs in the rest of the paper) as it is the case for 4 join FDs in the join result (in grey). The approximation degree $\alpha$ means that $\alpha$ tuples violate the exact FD. In the case of the AFD  ${\small\texttt{expire\_flag}} \approxi{1} {\small\texttt{dod}}$ in {\small\texttt{PATIENT}}  table, removing one tuple, either patient \#257 or patient \#250 will make the AFD become exact. When the two tables are joined, patient \#257 is removed due to the absence of the corresponding  ${\small\texttt{subject\_id}}$ value in  the {\small\texttt{ADMISSION}} table. The same phenomenon occurs for the 3 AFDs with approximation degree 1 discovered from  the {\small\texttt{ADMISSION}} table, namely:  ${\small\texttt{diagnosis}} \approxi{1} {\small\texttt{subject\_id}}$,  ${\small\texttt{diagnosis}} \approxi{1} {\small\texttt{insurance}}$,  and ${\small\texttt{diagnosis}} \approxi{1} {\small\texttt{admission\_location}}$. We note that the 5 remaining exact FDs from the \texttt{ADMISSION} table are reduced to minimal exact FDs in the set of FDs discovered from the join result. 
Finally, 10 exact FDs (in orange) that hold over the entire join result have to be discovered from the join result. However, if we join the two tables partially, only with the following combinations of tuples: [(\#249,\#252) or (\#249,\#251)] and [(\#250,\#251) or (\#250, \#252)], we can obtain the remaining 10 join FDs without the full join computation. 
Note that finding semantically correct FDs is orthogonal to our work and still an open problem for FD discovery from single tables: e.g., ${\small\texttt{dob}} \to {\small\texttt{dod}}$ discovered from  {\small\texttt{PATIENT}} and ${\small\texttt{admittime}} \to {\small\texttt{diagnosis}}$ from  {\small\texttt{ADMISSION}}, both also present in the join result, are not semantically correct. All FDs (from single or multiple tables) must be validated by domain experts before they can be used by downstream applications. Similarly,  AFDs can be either semantically incorrect (e.g., ${\small\texttt{diagnosis}} \to {\small\texttt{insurance}}$) or correct (e.g., ${\small\texttt{expire\_flag}} \to {\small\texttt{dod}}$) and  some can only surface as exact FDs in a join result due to errors in the single tables. Moreover, we would like to point out that these observations can be make regardless of the NULL semantics. A  better understanding of the mechanisms underlying the appearance of FDs in multi-table settings is hence needed. Based on these observations, the questions that motivated our work were the following: Instead of fully computing the join result before FD discovery, can we infer most of the FDs and accelerate the overall computation? Do we need to compute the full join anyway? How can we select the necessary tuples and attributes in a principled way?

The rest of the paper will attempt to answer these questions and propose efficient solutions as parts of our \jedi framework for discovering join FDs  from multiple relational tables.

\section{Preliminaries}
\label{sec:preliminaries}

Next, we recall the necessary definitions of FDs and join operators with their application to our problem.
\begin{definition}[Functional dependency satisfaction]
  Let $I$ be an instance over a relation $R$, 
  and $\bm{X}$, $\bm{Y}$ be two sets of attributes from $R$. $I$ satisfies a functional dependency $d: \bm{X} \to \bm{Y}$, denoted by  $I \models d$ if and only if:
  \begin{equation}
  \forall t_1,t_2\in I, t_1[\bm{X}] = t_2[\bm{X}] \Rightarrow t_1[\bm{Y}] = t_2[\bm{Y}].      
  \end{equation}
 \end{definition}
\begin{definition}[Logical implication between FDs]
  Let $R$ be a relation, 
  and $\bm{X}$, $\bm{Y}$, $\bm{Z}$ be three sets of attributes from $R$.
  Then a functional dependency $d: \bm{X} \to \bm{Z}$ logically implies a functional dependency $d^\prime: \bm{Y} \to \bm{Z}$  if, and only if, for every instance $I$ over $R$:
    \begin{equation}
    I \models d \text{ implies that } I \models d^\prime.
    \end{equation}

 \end{definition}
\begin{definition}[Approximate functional dependency] The approximate functional dependency (AFD), also called partial FD, denoted $a: \bm{X} \approxfd \bm{Y}$ is a functional dependency between $\bm{X}$ and $\bm{Y}$ that holds if and only if the minimum fraction of tuples  at most equal to $\varepsilon$ is removed from the relation  such that $\bm{X} \rightarrow \bm{Y}$  can hold. \\
The minimal fraction of tuples is computed as
\begin{equation}
e(\bm{X} \rightarrow \bm{Y})=1 - \sum_{c\in\pi_{\bm{X}}} \frac{\max\left\{c'| c' \in \pi_{\bm{X}\cup{\bm{Y}}} \text{and } c' \subseteq c\right\}}{|R|}
\end{equation}
with $\pi_{\bm{X}}$, the set of equivalence classes defining a partition of $R$ under $\bm{X}$, %
$\pi_{\bm{X}} = \left\{[t]_{\bm{X}} | t \in R\right\}$. The equivalence class of a tuple $t \in R$ with respect to a given set $\bm{X} \subseteq R$ is $[t]_{\bm{X}}=\left\{u \in R | t[A]=u[A] \forall A \in \bm{X}\right\}.$

\end{definition}

Next, we provide the definition of the \emph{natural join}. 
Let us consider \leftSch and \rightSch, two relations with at least one common attribute; \leftInst and \rightInst, two instances over the relations \leftSch and \rightSch, respectively;  $atts()$, a function taking an relational instance as input and returning its set of attributes. The set of common attributes between \leftSch and \rightSch is denoted  $X = atts(\leftInst) \bigcap atts(\rightInst)$. 

\begin{definition}[Natural join]
The natural join between \leftInst and \rightInst, denoted by $\leftInst \bowtie \rightInst$, is the instance such that
\setlist{nolistsep}
\begin{itemize}[noitemsep]
    \item $\forall t \in \leftInst \bowtie \rightInst, t\in (\leftInst\times \rightInst)$
    \item $\forall t \in \leftInst \bowtie \rightInst$ there exists two tuples $t_\leftInst\in\leftInst$ and $t_\rightInst\in\rightInst$
    such that $(\pi_{atts(\leftInst)} (t)=t_\leftInst) \wedge (\pi_{atts(\rightInst)} (t)=t_\rightInst) \wedge (\pi_{X} (t_\leftInst) = \pi_{X} (t_\rightInst)).$
    
    \end{itemize} 

\end{definition}
We recall the definitions of left and right semi-joins. 
\begin{definition}[Semi-join]
The left and right semi-joins between \leftInst and \rightInst, denoted by $\leftInst \ltimes \rightInst$ and $\leftInst \rtimes \rightInst$ respectively, are the instances such that $
    \leftInst \ltimes \rightInst = \pi_{atts(\leftInst)}(\leftInst \bowtie \rightInst)  $ %
    and $\leftInst \rtimes \rightInst = \pi_{atts(\rightInst)}(\leftInst \bowtie \rightInst).$

\end{definition}

In the case of relational instances without common attributes, the \emph{equi-join} operator can be used.  The \emph{equi-join} operator is defined as follows.
\begin{definition}[Equi-join]
Let $a\in atts(\leftInst)$ and $b\in atts(\rightInst)$ be two attributes. 
The equi-join between \leftInst and \rightInst on $a = b$, denoted by $\leftInst \bowtie_{a = b} \rightInst$, is the instance such that: 
    $\forall t \in \leftInst \bowtie_{a = b} \rightInst$, there exists two tuples $t_\leftInst\in\leftInst$ and $t_\rightInst\in\rightInst$
    such that: $\pi_{atts(\leftInst)} (t)=t_\leftInst \wedge \pi_{atts(\rightInst)} (t)=t_\rightInst \wedge \pi_{a} (t_\leftInst)\ =\ \pi_{b} (t_\rightInst).$
\end{definition}
Finally, we recall the definitions of left and right outer joins as follows.
\begin{definition}[Left outer join]
Let $\nu$ represent the null value. The left outer join between \leftInst and \rightInst, denoted by $\leftInst \leftouterjoin \rightInst$, is the instance such that $
    \leftInst \leftouterjoin \rightInst = \{ t |
                                                t\in (\leftInst \bowtie \rightInst) \cup (\leftInst\setminus \pi_{atts(\leftInst)}(\leftInst \bowtie \rightInst)) \times (\nu,\dots,\nu)
                                            \}.$
\end{definition}
\begin{definition}[Right outer join]
The right outer join between \leftInst and \rightInst, denoted by $\leftInst \rightouterjoin \rightInst$, is the instance such that: 
\begin{align*}
    \leftInst \rightouterjoin \rightInst = \{ t |
                                                t\in (\leftInst \bowtie \rightInst) \cup (\nu,\dots,\nu) \times (\rightInst\setminus \pi_{atts(\rightInst)}(\leftInst \bowtie \rightInst))
                                            \}.
\end{align*}
\end{definition}

From the previous definitions, we can define the notion of full outer join as follows.
\begin{definition}[Full outer join]
The right outer join between \leftInst and \rightInst, denoted by $\leftInst \fullouterjoin \rightInst$, is the instance such that:
\begin{align*}
    \leftInst \fullouterjoin \rightInst = (\leftInst \leftouterjoin \rightInst) \cup (\leftInst \rightouterjoin \rightInst).
\end{align*}
\end{definition}

 \begin{figure}[t]
	\centering
	\includegraphics[width=.9\linewidth]{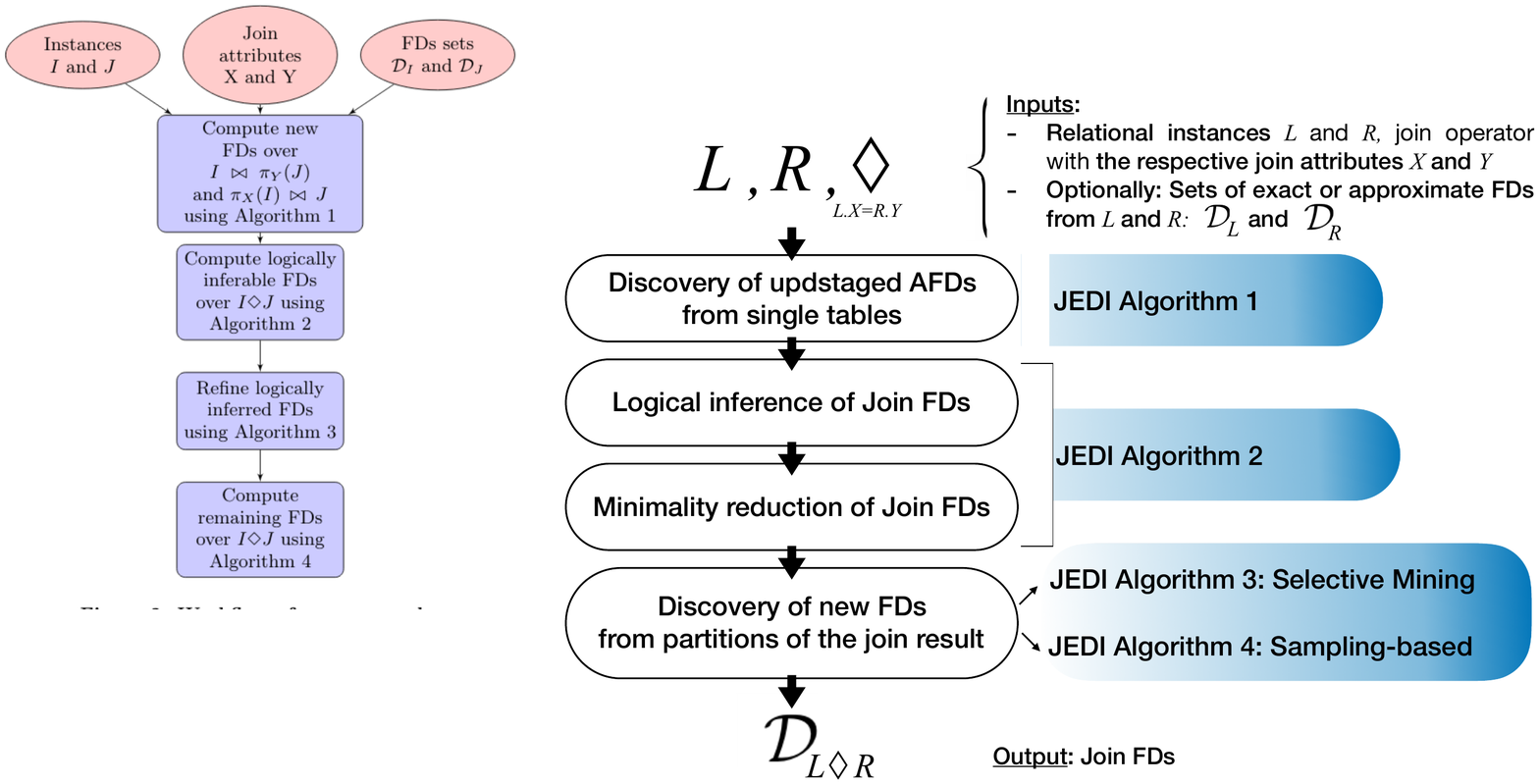}
	\caption{Workflow of \jedi for discovering join FDs from  $R \Diamond L$ with $\Diamond \in \{\bowtie;\ltimes;\rtimes;\fullouterjoin;\leftouterjoin;\rightouterjoin\}$}\label{fig:workflow}
\end{figure} 

\section{Efficient Discovery of Join FDs}
\label{sec:contributions}
In this section, we describe the problem we address and explain how we can infer and compute FDs that remain valid through a join operation with frugal computation.
First, we focus on the FDs discovered from single relations separately,  then we consider the FDs between attributes coming from the two (or more) joined relations %
and propose efficient methods to discover the exact FDs that cannot be inferred.

{\bf Problem statement}.\label{prob:generalJfds}
Let $\leftInst$ and $\rightInst$ be two instances of relations $\leftSch$ and $\rightSch$, respectively;  $\fdSet_\leftInst$ and $\fdSet_\rightInst$ the respective FD sets of the relational instances such that $\leftInst \models \fdSet_\leftInst$ and $\rightInst \models \fdSet_\rightInst$; $\Diamond \in \{\bowtie;\ltimes;\rtimes; \fullouterjoin ;\leftouterjoin ; \rightouterjoin\}$, a join operator;  and $\leftInst \Diamond_{\leftJoinAtt=\rightJoinAtt} \rightInst$ the result of joining $\leftInst$ and $\rightInst$ over the sets of attributes $\leftInst.\leftJoinAtt$ and $\rightInst.\rightJoinAtt$ with the operator $\Diamond$. The resolution of the \textsf{Join FD Discovery} problem aims at producing a set of functional dependencies, denoted $\fdSet_{\leftInst \Diamond_{\leftJoinAtt=\rightJoinAtt} \rightInst}$ such that:
\begin{equation}
    \fdSet_{\leftInst \Diamond_{\leftJoinAtt=\rightJoinAtt} \rightInst} = \fdSet^{new} \cup (\fdSet_{\leftInst} \cup \fdSet_{\rightInst}) \setminus \fdSet^{violated}
\end{equation}
where:
\begin{align}
    \fdSet^{new} &= \{d | (\leftInst \Diamond_{\leftJoinAtt=\rightJoinAtt} \rightInst \models d) \wedge (\leftInst \not\models d) \wedge (\rightInst \not\models d)\}\\[0.1cm]
    \fdSet^{violated} &= \{d \in \fdSet_{\leftInst} \cup \fdSet_{\rightInst} | \leftInst \Diamond_{\leftJoinAtt=\rightJoinAtt} \rightInst \not\models d \}
\end{align}

In the following theorem, we show that if the join operation is performed between instances then every FD over the initial instances stands over the joined instance, i.e., that $\fdSet^{violated}$ is empty.

\begin{theorem}%
  Let $\leftInst$ and $\rightInst$ be two instances over relations $\leftSch$ and $\rightSch$, respectively;  $\fdSet_\leftInst$ and $\fdSet_\rightInst$ be the two sets of all FDs such that $\leftInst \models \fdSet_\leftInst$ and $\rightInst \models \fdSet_\rightInst$, respectively. Then, we have:
  \begin{equation}
    \forall d \in \fdSet_\leftInst \cup \fdSet_\rightInst, {\leftInst \Diamond_{\leftJoinAtt=\rightJoinAtt} \rightInst} \models d
  \end{equation}
Equivalently: 
\begin{equation}
    {\leftInst \Diamond_{\leftJoinAtt=\rightJoinAtt} \rightInst} \models \fdSet_\leftInst \cup \fdSet_\rightInst
\end{equation}
\end{theorem}
Note that this theorem also holds in the presence of null values. Given the semantics of null values, every null value can be considered as a particular constant value (i.e., when null values are considered as identical,  all null values can be handled as one single constant value). In the rest of paper, we apply our algorithms regardless of the null semantics.

{\bf Our Solution. } 
In order to compute the set of FDs over $\leftInst \Diamond_{\leftJoinAtt=\rightJoinAtt} \rightInst$, we propose the workflow illustrated in Figure~\ref{fig:workflow}.
It consists in three steps corresponding to: (1) Discovery of approximate single-table FDs that are upstaged and become exact FDs via the join operation (\jedi step 1); (2) Logical inference and minimality reduction of join FDs (\jedi step 2); and (3) Computation of the remaining join FDs from partial join over selective mining or sampling (\jedi step 3). The steps are detailed in the next sections.

\subsection{From AFDs to Exact FDs Through Join}
New FDs may appear mechanically due to the join operation when tuples from one table cannot be joined with their counterpart in the other table, i.e., when some join attribute values are missing in one of the tables.  In certain cases, FDs that were approximate in a single table become exact in the join result. This mechanism is expressed more formally in the following theorem:

\begin{theorem}[Join FDs from Upstaged AFDs]
 Let $\leftInst$ and $\rightInst$ be two instances over relations $\leftSch$ and $\rightSch$, respectively, and $\fdSet_\leftInst$ and $\fdSet_\rightInst$ be the two sets of all FDs such that $\leftInst \models \fdSet_\leftInst$ and $\rightInst \models \fdSet_\rightInst$, respectively. Then the sets of upstaged FDs denoted $\fdSet^{new}_{{\leftInst}}$ and $\fdSet^{new}_{{\rightInst}}$ are the sets:
  \begin{align}
      \fdSet^{new}_{{\leftInst}} &= \{ d\ |\ d\not\in\fdSet_{\leftInst} \wedge \left(\leftInst \Diamond_{\leftJoinAtt=\rightJoinAtt} (\pi_{\rightJoinAtt}(\rightInst))\models d \right)\} \\
    \fdSet^{new}_{{\rightInst}} &= \{ d\ |\ d\not\in\fdSet_{\rightInst} \wedge \left(\pi_{\leftJoinAtt}(\leftInst) \Diamond_{\leftJoinAtt=\rightJoinAtt} \rightInst)\models d \right)\}
  \end{align}
\end{theorem}

\begin{example}
To illustrate the case of upstaged AFDs using the example of Figure 1, let us consider the NULL semantics where all null values are identical. The approximate FD ${\small\texttt{expire\_flag}}  \approxi{1} {\small\texttt{dod}}$ in table {\small\texttt{PATIENT}}  has an approximation degree of 1 as only one patient (\#257) violates the FD. However, in the join result of ${\small\texttt{PATIENT}} \bowtie_{{\small\texttt{subject\_id}}} {\small\texttt{ADMISSION}}$, the violating tuple \#257 has no counterpart in the {\small\texttt{ADMISSION}} table and it disappears from the join result. Consequently, the FD  ${\small\texttt{expire\_flag}} \to {\small\texttt{dod}}$ becomes exact in the join result. 

\begin{algorithm}[!h]
\SetAlgoLined
\KwInput{$\leftInst$ and $\rightInst$, two relational instances;\\ 
\hspace{1.1cm}$\leftJoinAtt$ and $\rightJoinAtt$, the sets of join attributes for $\leftInst$ and\\ \hspace{1.3cm} $\rightInst$ respectively; \\ 
\hspace{1.1cm}$\Diamond \in \{\bowtie;\ltimes;\rtimes;\fullouterjoin;\leftouterjoin;\rightouterjoin\}$ the join operator. \\
 {\bf Optional Input :} $\fdSet_{\leftInst}$ and $\fdSet_{\rightInst}$, the sets of AFDs or exact FDs over $\leftInst$ and $\rightInst$ respectively }
\KwResult{the FD sets $\fdSet^{up}_{\leftInst}$ and $\fdSet^{up}_{\rightInst}$}

\BlankLine
\SetKwFunction{procFDs}{upstagedFDs}
\SetKwFunction{procAFDs}{upstagedAFDs}
$\fdSet^{up}_{\leftInst}, \fdSet^{up}_{\rightInst} \gets \emptyset$ \;

\For{each pair of instances $(i,j)\in\{(L,R);(R,L)\}$}
{
    \uIf{ $\fdSet_{i}$ is not provided}
    { 
         $\fdSet_{i} \gets \texttt{computeFDs}(i) $ \;
    } 
    
    \uIf{ $\fdSet_{i}$ contains AFDs}
    {
        $\fdSet^{up}_{i} \gets$
        \procAFDs{$i$, $j$, $X$, $Y$, $\fdSet_{i}$} \; 
    } 
    
    \Else{
     $\fdSet^{up}_{i} \gets$
    \procFDs{$i$, $j$, $X$, $Y$, $\fdSet_{i}$}; } 
}  
\KwRet{($\fdSet^{up}_{\leftInst}$,$\fdSet^{up}_{\rightInst}$)}\;
\SetKwProg{myproc}{Subroutine}{}{}
\myproc{\procFDs{$I$,$J$,$X$,$Y$,$\fdSet$}}{
$\fdSet_{out} \gets \emptyset$\;
 $I_{join} \gets I \Diamond_{X=Y} (\pi_{Y}(J))$\;\label{line:joinLeft}
\If{$\fun{size}(I_{join}) < \fun{size}(I)$\label{line:sizeLeft}} 
    {
    $\fdSet_{cand} \gets$ generate candidate FDs for first level of $I_{join}$\;
    \Repeat{$\fdSet_{cand} = \emptyset$}
    {
         \prune FDs in $\fdSet_{cand}$ logically implied by FDs in $\fdSet_{out}$\;\label{line:pruneTane} %
        \prune FDs in $\fdSet_{cand}$ logically implied by FDs in $\fdSet$\;\label{line:pruneAlreadyKnown}  %
         \add to $\fdSet_{out}$ the FDs from $\fdSet_{cand}$ holding in $I$\;
        $\fdSet_{cand} \gets$ generate candidate FDs for next level\;
    }
}
\KwRet $\fdSet_{out}$}
\BlankLine
\SetKwProg{myproc}{Subroutine}{}{}
\myproc{\procAFDs{$I$,$J$,$X$,$Y$,$\fdSet$}}{
$\fdSet_{out} \gets \emptyset$\;
\For{$d \in \fdSet$}
{
    $V \gets$ get violating tuples for $d$ in $I$ \;
    \If{$V \Diamond_{X=Y} (\pi_{Y}(J)) = \emptyset$\label{line:pruneAFDs}}
    {
       \add $d$ to $\fdSet_{out}$\;
    }
}
\prune FDs in $\fdSet_{out}$ implied by each others  %
\KwRet $\fdSet_{out}$}
\caption{Find Join FDs from Upstaged AFDs}
\label{alg:filteringFDs}
\end{algorithm}

\end{example}

To compute join FDs from upstaged AFDs, we propose Algorithm~\ref{alg:filteringFDs}. Lines 2--10 handle the inputs of the user if s/he can provide FDs and/or AFDs for each table participating in the join operation. If not, exact FDs are computed from each single table. For each side of the join, if exact FDs have been provided, the subroutine \compUpstaged is executed (lines\#~8 and 12) and computes partially the join only with the join attributes from the left side table (line\#~\ref{line:joinLeft}) to  check the assumption of the join value set preservation \cite{booksGarcia}. If the assumption is violated (i.e., if some tuples have been deleted through the join operation (line\#~\ref{line:sizeLeft})), some upstaged join FDs are produced. 
The subroutine discovers the FDs in the input instance by taking into account the previously discovered FDs and improves the pruning (line\#~\ref{line:pruneAlreadyKnown} of \compUpstaged). On the other hand, if AFDs are provided as inputs of Algorithm~\ref{alg:filteringFDs} (lines \#~5--6 and 25), the subroutine \texttt{upstagedAFDs}  will check for each AFD if the join of its set of violating tuples with the instance from the other side of the join leads to an empty instance (line\#~\ref{line:pruneAFDs}). In this case, the AFD becomes exact in the joined instance and thus is added to the output set of exact FDs.  %
In this algorithm, the computation is performed over one table at a time, and not over the complete join result. Thus, the FD discovery focuses on FDs whose attributes belong to one joined table only. Next, we discover the FDs containing attributes from both instances by relying on the characteristics of the join and on the FDs discovered previously. 

\subsection{Exact FDs from Joined Tables}

Other join FDs are FDs containing attributes from the result of a join operation between multiple tables. Compared to the previous join FDs from upstaged AFDs, they will include a mix of attributes coming from each table participating in the join. Their definition is formalized as follows. 

\begin{definition}%
Let $\leftInst$ and $\rightInst$ be two instances over relations $\leftSch$ and $\rightSch$, respectively. 
  An FD $d$ is said to be multi-table and specific to the joined instance $\leftInst \Diamond_{\leftJoinAtt=\rightJoinAtt} \rightInst$ if:
      $d$ holds in $\leftInst \Diamond_{\leftJoinAtt=\rightJoinAtt} \rightInst$; and 
      $d$ contains at least one attribute occurring in $atts(\leftInst)\setminus \leftJoinAtt$ and one attribute occurring in $atts(\rightInst)\setminus \rightJoinAtt$.
\end{definition}
\begin{example}In our example of Figure 1, FDs specific to the joined result are highlighted in orange.
For example,  {\small\texttt{gender, h\_expire\_flag}}$\to$ {\small\texttt{insurance}} is specific to the join of {\small\texttt{Patient}} and {\small\texttt{Admission}}. It holds in {\small\texttt{PATIENT}} $\bowtie$ {\small\texttt{ADMISSION}}. Attributes {\small\texttt{gender}} and {\small\texttt{expire\_flag}} come from  {\small\texttt{PATIENT}} and attribute {\small\texttt{insurance}} comes from {\small\texttt{ADMISSION}}.
\end{example}

Next, we consider the case of FDs that have all attributes in their left-hand side coming from only one relational instance and we prove several interesting properties about these logically inferable FDs.
Then, we examine the properties of multi-table join FDs with left-hand side attributes obtained from both initial instances.

\subsubsection{Logically inferable Join FDs}
We will now show that the join FDs with left-hand side (\texttt{lhs}) attributes coming from only one initial instance can be deduced from the sets of FDs over the initial instances. 
To prove the theorem, we first state the following lemma showing that an FD with \texttt{lhs} attributes coming from only one single instance cannot exist if their right-hand side (\texttt{rhs}) is not functionally defined by the set of join attributes.

\begin{theorem}\label{thm:notModelsImplication}
Let $\leftInst$ and $\rightInst$ be two instances over relations $\leftSch$ and $\rightSch$, respectively. Let $\leftInst \Diamond_{\leftJoinAtt=\rightJoinAtt} \rightInst$ be a join result with $\leftJoinAtt \subseteq atts(\leftInst)$, $\rightJoinAtt  \subseteq atts(\rightInst)$. For all $A  \subseteq atts(\leftInst)\setminus \leftJoinAtt$ and $B \subseteq atts(\rightInst)\setminus \rightJoinAtt$:
$$\text{if }\leftInst \Diamond_{\leftJoinAtt=\rightJoinAtt} \rightInst \not\models \leftJoinAtt \to B \text{ then } \leftInst \Diamond_{\leftJoinAtt=\rightJoinAtt} \rightInst \not\models A \to B $$
\end{theorem}

\begin{proof}[sketch]
Let $x_1,\dots,x_n$ being values over the attributes in $\leftJoinAtt$, and $b_1,\dots,b_m,b^\prime_1,\dots,b^\prime_m$ being values over the attributes in $B$. If $\leftInst \Diamond_{\leftJoinAtt=\rightJoinAtt} \rightInst \not\models \leftJoinAtt \to B $ then there exist two tuples:
\begin{align*}
    &t_\rightInst(x_1,\dots,x_n,b_1,\dots,b_m,\dots)\\
    &t_\rightInst^\prime(x_1,\dots,x_n,b^\prime_1,\dots,b^\prime_m,\dots)
\end{align*}

in $\rightInst$ such that:
    $\exists i \in [1,\dots,m], b_i \neq b^\prime_i$ and
  there exists a tuple $t_\leftInst(x_1,\dots,x_n,a_1,\dots,a_k,\dots)$ in $\leftInst$ with $a_1,\dots,a_k$, the values of the attributes in $A$ (otherwise, tuples $t_\rightInst$ and $t_\rightInst^\prime$ would have been filtered during the join operation).
Thus, the join $\leftInst \Diamond_{\leftJoinAtt=\rightJoinAtt} \rightInst$  leads to the two tuples:
\[
\begin{split}
    t(x_1,\dots,x_n,a_1,\dots,a_k,b_1,\dots,b_m,\dots)\\ t^\prime(x_1,\dots,x_n,a_1,\dots,a_k,b^\prime_1,\dots,b^\prime_m,\dots)
\end{split}
\]
which violate the FD $A \to B$ and 
$\leftInst \Diamond_{\leftJoinAtt=\rightJoinAtt} \rightInst \not\models A \to B $\end{proof}

\begin{example}
    To illustrate the property proved in Theorem~\ref{thm:notModelsImplication}, we observe that the diagnosis is not determined by the patient identifier in Figure 1, for example patient \#249 has been admitted three times for a different pathology each time, i.e. : ${\small\texttt{PATIENT}}\bowtie_{ {\small\texttt{subject\_id}}} {\small\texttt{ADMISSION}} \not\models  {\small\texttt{subject\_id}} \to  {\small\texttt{diagnosis}}$. 
    From Theorem~\ref{thm:notModelsImplication}, we know that {\small\texttt{diagnosis}} in the join result ${\small\texttt{PATIENT}}\bowtie_{{\small\texttt{subject\_id}}} {\small\texttt{ADMISSION}}$ cannot be determined by any set of attributes coming from \texttt{PATIENT} table. Such similar inferences may be trivial for the user, but they usually require the knowledge of the attribute semantics. If not encoded, they are difficult  to capture by a system. However, the property proved in Lemma~\ref{thm:notModelsImplication} can be used to drastically reduce the set of possible FDs that can appear after a join operation.
\end{example}

From our proof of Theorem~\ref{thm:notModelsImplication},  we can characterize a subset of the set of FDs with \texttt{lhs} attributes coming from only one initial instance and that hold in the joined instance.
\begin{theorem}\label{thm:MonoInstanceCrossJoinFDs}
  Let $\leftInst$ and $\rightInst$ be two instances over relations $\leftSch$ and $\rightSch$, respectively. 
Let $\leftInst \Diamond_{\leftJoinAtt=\rightJoinAtt} \rightInst$ be a join result with $\leftJoinAtt \subseteq atts(\leftInst)$, $\rightJoinAtt  \subseteq atts(\rightInst)$. 
For all $A  \subseteq atts(\leftInst)\setminus \leftJoinAtt$ and $B \subseteq atts(\rightInst)\setminus \rightJoinAtt$,\\
 If $\leftInst \Diamond_{\leftJoinAtt=\rightJoinAtt} \rightInst \models A \to \leftJoinAtt \wedge  \leftInst \Diamond_{\leftJoinAtt=\rightJoinAtt} \rightInst \models \leftJoinAtt \to B$,\\
 Then $\leftInst \Diamond_{\leftJoinAtt=\rightJoinAtt} \rightInst \models A \to B$.
\end{theorem}
\begin{proof}
This is trivially proved by transitivity, with the use of Armstrong's transitivity axiom.
\end{proof}
\begin{example}

In ${\small\texttt{PATIENT}} \bowtie_{{\small\texttt{ subject\_id}}}{\small\texttt{ADMISSION}}$ result illustrated in Figure 1, we observe that the diagnosis determines the date of birth, i.e., 
${\small\texttt{diagnosis}} \to {\small\texttt{dob}}.$
The reason is that we have:  
${\small\texttt{admission\_location}}$, ${\small\texttt{diagnosis}} \to {\small\texttt{subject\_id}}$ in ${\small\texttt{ADMISSION}}$  and ${\small\texttt{subject\_id}} \to {\small\texttt{dob}}$ in ${\small\texttt{PATIENT}}$.  
Since these tables do not contain any null values for the \texttt{lhs} and \texttt{rhs} attributes, joining them with attribute ${\small\texttt{subject\_id}}$ leaves these FDs unchanged with no violation. By transitivity, we obtain:
${\small\texttt{diagnosis}} \to {\small\texttt{dob}}.$
\end{example}

\begin{algorithm}[!h]
\SetAlgoLined
\KwInput{$\leftInst$ and $\rightInst$, two instances; \\ 
\hspace{1.1cm}$\leftJoinAtt$ and $\rightJoinAtt$, the sets of join attributes for $\leftInst$ and\\ 
\hspace{1.3cm} $\rightInst$ respectively;\\
\hspace{1.1cm}$\fdSet_\leftInst$ and $\fdSet_\rightInst$, exact FD sets from  $\leftInst$ and $\rightInst$;\\ 
\hspace{1.1cm}$\Diamond \in \{\bowtie;\ltimes;\rtimes;\fullouterjoin;\leftouterjoin;\rightouterjoin\}$ a join operator.}
\KwResult{$\fdSet_{out}$, the FDs inferred over ${\leftInst \Diamond_{\leftJoinAtt=\rightJoinAtt} \rightInst}$}

\SetKwFunction{proc}{infer}
\SetKwFunction{procRef}{refine}
$\fdSet_{inf} \gets$ \proc{$\leftJoinAtt$,$\rightJoinAtt$,$\fdSet_{{\leftInst}}$,$\fdSet_{{\rightInst}}$}\;
$\fdSet_{inf} \gets \fdSet_{inf} \cup$ \proc{$\rightJoinAtt$,$\leftJoinAtt$,$\fdSet_{{\rightInst}}$,$\fdSet_{{\leftInst}}$}\;
\KwRet{\procRef{$\leftInst$,$\rightInst$,$X$,$Y$,$\fdSet_{inf}$, $\Diamond$}}\;
\BlankLine
\setcounter{AlgoLine}{0}
\SetKwProg{myproc}{Subroutine}{}{}
\myproc{\proc{$X$,$Y$,$\fdSet$, $\fdSet^\prime$}}{
$\fdSet_{out} \gets \emptyset$\;
\ForAll{$A \to X$ in $\fdSet$\label{line:transX}}
{
    \ForAll{$Y \to b$ in $\fdSet^\prime$\label{line:transY}}
    {
       add $A \to b$ to $\fdSet_{out}$\;
    }
}
\KwRet $\fdSet_{out}$\;}

\BlankLine
\SetKwProg{myproc}{Subroutine}{}{}
\myproc{\procRef{$\leftInst$,$\rightInst$,$X$,$Y$,$\fdSet_{inf}$, $\Diamond$}}{
Let $\fdSet_{out} \gets \fdSet_{inf}$\;
\ForAll{$A \to b$ in $\fdSet_{inf}$\label{line:fdRef}}
{
    Let $I \gets \pi_{X \cup A}(\leftInst) \Diamond \pi_{Y \cup \{b\}}(\rightInst)$\label{line:instRef}\;
    \ForAll{$A^\prime \subset A$\label{line:fdSub1}}
    {
        \If{$A^\prime \to b$ holds in $I$\label{line:testFdHolds}\label{line:fdSub2}}
        {
        add $A^\prime \to b$ to $\fdSet_{out}$\;
        remove FDs implied by $A^\prime \to b$ from $\fdSet_{out}$\;
        }
    }
}
 \KwRet $\fdSet_{out}$\;}
\caption{Infer Join FDs}
\label{algo:MonoInst_CrossJoinFDs}
\end{algorithm}

To compute the set of FDs with \texttt{lhs} attributes coming from a single instance, as described in Theorem~\ref{thm:MonoInstanceCrossJoinFDs}, we propose Algorithm~\ref{algo:MonoInst_CrossJoinFDs}.
First, the subroutine \subInfer extracts the FDs that can be retrieved by transitivity (lines\#~\ref{line:transX} and~\ref{line:transY} in subroutine \subInfer). Note that in the case of equijoins, equality of values might be enforced between sets of attributes with different names (i.e., $X$ and $Y$ might be different), thus for the general case, the FD (line\#~\ref{line:transY}) cannot be simplified into an FD $X\to b$. At the opposite, if we restrict our join operations to natural joins only, such a simplification can be made.

Then, for each FD returned by \subInfer, 
the subroutine \subRefine checks whether the FD is minimal or if a subset of its \texttt{lhs} leads to a minimal FD.
To do so, subroutine \subRefine uses an horizontal partition of the joined instances in which only the necessary attributes to perform the verification are considered (line\#~\ref{line:instRef}).
These necessary attributes are the join attributes (to perform the join operation), and the \texttt{lhs} and \texttt{rhs} attributes of the refined FD $A \to b$ (line\#~\ref{line:fdRef}) as \subRefine only considers candidates with subsets of $A$ as \texttt{lhs} and $b$ as \texttt{rhs} (lines\#~\ref{line:fdSub1} and \ref{line:fdSub2}).

\subsubsection{FDs with \texttt{lhs} attributes from multiple tables }

Now, we characterize the set of FDs which hold on a join result and such that their \texttt{lhs} attributes come from both initial instances.
We characterize two kinds of FDs with multi-table attributes in \texttt{lhs}: (1) FDs that can be deduced directly using a simple logical reasoning, and (2)  FDs that need to be discovered and validated from the data. 
For example, {\small\texttt{gender, expire\_flag}}$ \to $ {\small\texttt{insurance}} of our  example has attributes from  {\small\texttt{PATIENT}} in \texttt{lhs} and attributes from  {\small\texttt{ADMISSION}} in \texttt{rhs} and it cannot be inferred logically. Other FDs with the same properties are illustrated in orange in Figure ~1. In the following theorem, we show that if \texttt{lhs} attributes of an FD come from the instances participating in the join, then we cannot predict their validity without checking them directly with some representative (if not all) tuples of the join result:
\begin{theorem}
Let $\leftInst$ and $\rightInst$ be two instances over relations $\leftSch$ and $\rightSch$, respectively. Let $\leftInst \Diamond_{\leftJoinAtt=\rightJoinAtt} \rightInst$ be a join result with $\leftJoinAtt \subseteq atts(\leftInst)$, $\rightJoinAtt  \subseteq atts(\rightInst)$. 
We cannot guarantee that all FDs over $\leftInst \Diamond_{\leftJoinAtt=\rightJoinAtt} \rightInst$ can be inferred from Armstrong's axioms over the FDs over $\leftInst$ and $\rightInst$ taken separately.
\end{theorem}
\begin{proof}
In the two following instances \leftInst and \rightInst, we can see that only the FDs $\rightJoinAtt A^\prime \to b$ and $\rightJoinAtt b \to A^\prime$ hold.
\begin{center}
    \begin{minipage}{0.25\columnwidth}
    \centering
        \begin{tabular}{cc}
            \multicolumn{2}{c}{$\leftInst$}\\
            \hline
            $\leftJoinAtt$ & $A$ \\
            \hline
            0 & 0\\
            1 & 0\\
            1 & 1\\
            2 & 2\\
            & \\
            & \\
        \end{tabular}
    \end{minipage}
    \begin{minipage}{0.25\columnwidth}
    \centering
        \begin{tabular}{ccc}
            \multicolumn{3}{c}{$\rightInst$}\\
            \hline
            $\rightJoinAtt$ & $A^\prime$ & $b$ \\
            \hline
            0 & 0 & 0\\
            1 & 0 & 0\\
            1 & 1 & 1\\
            2 & 1 & 0 \\
               & \\
            & \\
        \end{tabular}
    \end{minipage}
    \begin{minipage}{0.4\columnwidth}
\centering
\begin{tabular}{cccc}
    \multicolumn{4}{c}{${\leftInst \Diamond_{\leftJoinAtt=\rightJoinAtt} \rightInst}$}\\
    \hline
    $\leftJoinAtt=\rightJoinAtt$ & $A$ & $A^\prime$ & $b$\\
    \hline
    0 & 0 & 0 & 0\\
    1 & 0 & 0 & 0\\
    1 & 0 & 1 & 1\\
    1 & 1 & 0 & 0\\
    1 & 1 & 1 & 1\\
    2 & 2 & 1 & 0\\
\end{tabular}
\end{minipage}
\end{center}
In the join result, the FD $AA^\prime \to b$ holds but it cannot be inferred using Armstrong's axioms over the FDs discovered from each instance $\leftInst$ and $\rightInst$.
\end{proof}
This theorem motivates the need for designing a new method for computing FDs from partial join results, as we cannot always infer all the FDs only using logical reasoning. However, we can rely on the following theorem to greatly reduce number of remaining FDs to check from the data:

\begin{theorem}\label{theo6}
 Let $\leftInst$ and $\rightInst$ be two instances over relations $\leftSch$ and $\rightSch$, respectively. 
Let $\leftInst \Diamond_{\leftJoinAtt=\rightJoinAtt} \rightInst$ be a join result with $\leftJoinAtt \subseteq atts(\leftInst)$, $\rightJoinAtt  \subseteq atts(\rightInst)$. 
For all $A  \subseteq atts(\leftInst)$, $A^\prime  \subseteq atts(\rightInst)$ and $b \in atts(\rightInst)$:
If $\leftInst \Diamond_{\leftJoinAtt=\rightJoinAtt} \rightInst \models AA^\prime \to b$, Then $\leftInst \Diamond_{\leftJoinAtt=\rightJoinAtt} \rightInst \models \rightJoinAtt A^\prime \to b$.
\end{theorem}

\begin{algorithm}[t]
\SetAlgoLined
\KwInput{$\leftInst$ and $\rightInst$, two relational instances; \\ 
\hspace{1.1cm}$\leftJoinAtt$ and $\rightJoinAtt$, the sets of join attributes for $\leftInst$ and \\ \hspace{1.3cm} $\rightInst$ respectively;  \\
\hspace{1.1cm}$\fdSet_\leftInst$ and $\fdSet_\rightInst$, exact FDs sets from $\leftInst$ and $\rightInst$; \\
\hspace{1.1cm}$\fdSet_{\leftInst\Diamond \rightInst}$, set of FDs inferred from ${\leftInst \Diamond_{\leftJoinAtt=\rightJoinAtt} \rightInst}$;  \\
\hspace{1.1cm}$\Diamond \in \{\bowtie;\ltimes;\rtimes; 
\fullouterjoin;\leftouterjoin;\rightouterjoin\}$ a join operator.}
\KwResult{$\fdSet_{final}$, the final set of FDs from ${\leftInst \Diamond_{\leftJoinAtt=\rightJoinAtt} \rightInst}$}
\SetKwFunction{proc}{discover}
$\fdSet_{out} \gets$ \proc{$\leftInst$,$\rightInst$,$\leftJoinAtt$,$\rightJoinAtt$,$\fdSet_{\rightInst}$,$\fdSet_{\leftInst\Diamond \rightInst}$}\;
$\fdSet_{final} \gets \fdSet_{out} \cup$ \proc{$\rightInst$,$\leftInst$,$\rightJoinAtt$,$\leftJoinAtt$,$\fdSet_{\leftInst}$,$\fdSet_{\leftInst\Diamond \rightInst}$}\;
\KwRet $\fdSet_{final}$
\BlankLine
\SetKwProg{myproc}{Subroutine}{}{}
\myproc{\proc{$I$,$J$,$X$,$Y$,$\fdSet_{J}$,$\fdSet_{\leftInst\Diamond \rightInst}$}}{
$\fdSet_{out} \gets \emptyset$\;
\ForAll{$Y \to b$ in $\fdSet_{{J}}$\label{line:joinAtt}}
{
    \ForAll{$A\subseteq atts(I)\setminus X$ \label{line:joinAttLhs}}
    {
    \If{$\not\exists A^\prime \subseteq A, A^\prime \to b \in \fdSet_{I\Diamond J}$ and $A \to b$ holds in $I \Diamond J$\label{line:joinAttLhsEnd}}
            {
                add $A \to b$ to $\fdSet_{out}$\;
            }
    }
}
\ForAll{$YA^\prime \to b \in \fdSet_{J}$ such that $A^\prime\not\to b$\label{line:joinAttAndOthers}}
    {
        \ForAll{$A \cup A^\prime \to b$ such that $A \subseteq atts(I)\setminus X$\label{line:joinAttAndOthersLhs}}
        {
            \If{$A \cup A^\prime \to b$ holds in $I \Diamond J$\label{line:joinAttAndOthersLhsEnd}}
            {
                add $A \cup A^\prime \to b$ to $\fdSet_{out}$\;
            }
        }
    }
\KwRet $\fdSet_{out}$\;
}
\caption{Discover Join FDs with selective mining}
\label{algo:MultiInst_CrossJoinFDs}
\end{algorithm}

In-line with Theorem~\ref{theo6}, we propose Algorithm 3 for selective mining and use the FDs previously discovered with Algorithms~\ref{alg:filteringFDs} and \ref{algo:MonoInst_CrossJoinFDs} to compute the remaining join FDs. Intuitively, Theorem~\ref{theo6} shows that a given attribute $b$ can be a \texttt{rhs} of a remaining join FDs only if we have previously found an FD of the form $YA \to b$ with $Y$ being the join attributes of the instance containing $b$.
Thus, it allows us to focus only on the plausible \texttt{rhs} (lines\#~\ref{line:joinAtt} and \ref{line:joinAttAndOthers} in subroutine \texttt{discover}) and explore their candidate \texttt{lhs} (lines\#~\ref{line:joinAttLhs}-\ref{line:joinAttLhsEnd} and \ref{line:joinAttAndOthersLhs}-\ref{line:joinAttAndOthersLhsEnd}).  In practice, there is no need to generate every candidate FDs initially. Instead, candidate FDs can be explored by generating a first level containing only the smallest candidates and by generating upper levels only when currently evaluated candidates are not valid. Moreover, we can avoid the computation of the full join by deleting a given \texttt{lhs} attribute $a$ if $a$ is not a possible \texttt{rhs} and, for every FD candidate $d: A\to b$ such that $a\in A$, $d$ is logically implied by previously discovered FDs.

\subsubsection{Sampling-based Discovery of Join FDs}
As the computation of the join FDs by Algorithm~\ref{algo:MultiInst_CrossJoinFDs} can be expensive both in time and space, we propose Algorithm~\ref{algo:statisticalInferenceJoinFDs} which relies on selective sampling. To this extent, Algorithm~\ref{algo:statisticalInferenceJoinFDs} starts with the \texttt{micro\_join} subroutine that compute FDs from micro-joins only between the tuples from $L$ and $R$ that are selected by the \texttt{selective\_sampling} subroutine (lines \#6 and \#19--25).  The set of tuple ids to join is generated in the \texttt{generate\_ids\_set} subroutine  (lines \#23-24, \#26-38). It builds a tree where the depth levels correspond to the order of the attributes  with the fewest number of distinct values as the first level and $n_v$ the last level to consider (line \#29); the nodes correspond to  distinct values per attribute (i.e., level) (line \#33). Then, if an attribute value corresponds to a unique tuple for a level less than $n_v$, the tuple id is selected as the representative tuple of a branch to be considered for the micro-join; if multiple tuples share the same value, $n_b$ tuples will be selected.
Then, for each sample, the instances are joined (line\#~\ref{line:joinPair}) and the FDs over the resulting instance are computed and added to  $\mathcal{E}_{out}$, the set of all computed FDs  (lines\#~\ref{line:computeFDsPairs} and \ref{line:addNewFDsPairs}).
Finally, subroutine \texttt{computeLvLFDs} (line\#~\ref{line:callPartition}) extracts from the set $\mathcal{E}_{out}$ the FDs that are valid in every sample. An FD $d$ can be considered as valid in a set of FD $\fdSet$ if either (1) there exists a logically equivalent FD in $\fdSet$ or (2) there exists an FD logically implying $d$ in $\fdSet$. However, we cannot guarantee that an FD discovered from samples is exact in the joined instance without the full exploration. Indeed, %
for a given FD $d$, in the worst case, %
 the projection of $\leftInst\bowtie\rightInst$ over the attributes in $d$  needs to be explored to check the non-existence of a pair of tuples violating $d$. We refer to \cite{davies1994np} for a formal proof of the NP-completeness of finding minimal FDs.
\begin{algorithm}[h]
\SetAlgoLined
\SetKwFunction{procmj}{micro\_join}
\SetKwFunction{procLvl}{compute\_FDs}
\SetKwFunction{procSelSamp}{selective\_sampling}
\SetKwFunction{procGenList}{generate\_ids\_set}

\KwInput{$\leftInst$ and $\rightInst$, two instances; \\ $\leftJoinAtt$ and $\rightJoinAtt$, the sets of join attributes for $\leftInst$ and 
 $\rightInst$ respectively;\\
$\fdSet_\leftInst$ and $\fdSet_\rightInst$, exact FDs sets of $\leftInst$ and $\rightInst$;  \\
$\Diamond \in \{\bowtie;\ltimes;\rtimes;\fullouterjoin;\leftouterjoin;\rightouterjoin\}$ a join operator;  \\ %
$n_b$ the number of representative tuples to pick per tree level; \\
and 
$n_v$ the number of tree levels to avoid in sample generation; \\ 
\KwResult{$\fdSet_{out}$, the FDs over ${\leftInst \Diamond_{\leftJoinAtt=\rightJoinAtt} \rightInst}$}

\BlankLine
$\mathcal{E}_\fdSet, \mathcal{P} \gets \emptyset$\;
$\mathcal{E}_\fdSet \gets \procmj(\leftInst,\rightInst, \leftJoinAtt, \rightJoinAtt,n_b, n_v)$ \;
$\fdSet_{out} \gets $ \procLvl{$\mathcal{E}_\fdSet$} \label{line:callPartition}\;
\KwRet $\fdSet_{out}$\;

\BlankLine
\SetKwProg{myproc}{Subroutine}{}{}
\myproc{\procmj{$\leftInst$, $\rightInst$, $\leftJoinAtt$, $\rightJoinAtt$, $n_{b}$, $n_{v}$}}{
    $ids \gets $ \texttt{selective\_sampling}($L,R,n_b,n_v$)\;%
    $\mathcal{P} \gets (\sigma_{\leftJoinAtt\in ids}(\leftInst),\sigma_{\rightJoinAtt\in ids}(\rightInst))$\;
   \ForAll{$(\leftInst_i,\rightInst_i) \in \mathcal{P}$}
{
    $J_i \gets \leftInst_i \Diamond \rightInst_i$\label{line:joinPair}\;
    $\fdSet_i \gets$ compute exact FDs over $J_i$\label{line:computeFDsPairs}\;
    {\bf add} $\fdSet_i$ to $\mathcal{E}_\fdSet$\label{line:addNewFDsPairs}\;
}
}

\KwRet $\mathcal{E}_\fdSet$\;

\BlankLine
\SetKwProg{myproc}{Subroutine}{}{}
\myproc{\procLvl{$\{\fdSet_1;\dots;\fdSet_n \}$}}{
$\fdSet_{out}\gets \emptyset$\;

\ForAll{$d \in \bigcup_{i\in [1,n]} \fdSet_i$}
{
    \If{$\forall \fdSet_i \in \{\fdSet_1;\dots;\fdSet_n\}, \exists d^\prime\in\fdSet_i, d^\prime \Rightarrow d$}
    {
        {\bf add} $d$ to $\fdSet_{out}$ \;
    }
}
\KwRet $\fdSet_{out}$\;
}
\BlankLine
\SetKwProg{myproc}{Subroutine}{}{}
\myproc{\procSelSamp{$L$, $R$, $X$, $Y$, $n_b$, $n_v$}}{
$\mathcal{E}_{ids} \gets \pi_{\leftJoinAtt}(\leftInst) \cap \pi_{\rightJoinAtt}(\rightInst)$\;
$\leftInst_{ids} \gets \sigma_{\leftJoinAtt \in \mathcal{E}_{ids}}(\leftInst)$\;
$\rightInst_{ids} \gets \sigma_{\rightJoinAtt \in \mathcal{E}_{ids}}(\rightInst)$\;
$\mathcal{E}_{ids}^\leftInst \gets$ \procGenList{$\leftInst_{ids}$, $\leftJoinAtt$, $n_b$, $n_v$}\;
$\mathcal{E}_{ids}^\rightInst \gets$ \procGenList{$\rightInst_{ids}$, $\rightJoinAtt$, $n_b$, $n_v$}\;}

\KwRet $\pi_\leftJoinAtt(\mathcal{E}_{ids}^\leftInst) \cap \pi_\rightJoinAtt(\mathcal{E}_{ids}^\rightInst)$\;
\BlankLine
\SetKwProg{proc}{Subroutine}{}{}
\proc{\procGenList{$I$, $A$, $n_b$, $n_v$}}{
$ids_{out} \gets \emptyset$\;
$I^\prime \gets I$\;
$nvList \gets$ list the attributes in $atts(I^\prime)\setminus A$ in ascending order of distinct values with size limit $n_v$\; 
        
\ForAll{$a \in nvList$} %
        {
            $I^\prime_a \gets \pi_a(I^\prime)$\;
            \ForAll{$v \in I^\prime_a$}
                {
                    $I^{\prime\prime} \gets \pi_A(\sigma_{a = v}(I^\prime))$\;
                    \uIf{$|I^{\prime\prime}| = 1 $}
                        {
                            add $I^{\prime\prime}$ to $ids_{out}$\;
                        }
                    \Else
                        {
                            add $n_b$ tuples from $I^{\prime\prime}$ to $ids_{out}$
                        }
                }
        }

\KwRet $ids_{out}$\;}

}
\caption{Sampling-based discovery of Join FDs}
\label{algo:statisticalInferenceJoinFDs}
\end{algorithm}

\subsection{The 4 C's}

\subsubsection{Coverage} We observed that the cardinalities and overlap of the join attribute values are rarely preserved through a join operation and this has a great impact on FD discovery depending on the join operator used. Various sets of join FDs can be computed and we summarized the different cases in Figure~\ref{fig:cardinalities} that recaps when \jedi algorithms are  executed. Our example corresponds to the last line and $\bowtie$ column of the table with $(0..N;0..N)$ cardinalities where the join attribute  {\small\texttt{subject\_id}} with value \#257 in  \texttt{PATIENT} has no counterpart in {\small\texttt{ADMISSION}} table; \#247, 248, and \#253 in {\small\texttt{ADMISSION}} have no counterpart in {\small\texttt{PATIENT}} table, and \#252 and \#249 in {\small\texttt{PATIENT}} are present multiple times in {\small\texttt{ADMISSION}} table. To quantify this phenomenon, we define the notion of join coverage and compute it as follows:
\begin{align*}
    Coverage(R\Diamond L) =
        \frac{1}{2}  \bigg( &Cov(R\Diamond L, L, X) + Cov(R\Diamond L, R, Y)\bigg)
    \\
    \text{with }Cov(Join, I, a) &= \frac{1} {|\pi_{a}(I) |} \sum\limits_{\forall v \in \pi_{a}(I)} 
                \frac{|\sigma_{a=v}(Join)) |}
                     {|\sigma_{a=v}(I) |}.
\end{align*}
  $X$ and $Y$ denote the join attributes of $\leftInst$ and $\rightInst$, respectively. $I$ is a considered instance and $a$ the considered join attribute. If $Coverage(R\Diamond L)=0$, no tuple from $L$ can be joined with tuples in $R$. For $Coverage(R\Diamond L)<1$, some tuples in $L$ (or $R$) may be missing from the join result, as it is the case for patients \#257 in {\small\texttt{PATIENT}} and \#247, \#248, and \#253 in {\small\texttt{ADMISSION}} that do not have their counterparts in the other table in our example. For $Coverage(R\Diamond L)=1$, there are as many tuples in both tables $L$ and $R$ as in the join result. For $Coverage(R\Diamond L)>1$, there are more tuples in the join result than in tables $L$ or $R$ as some tuples may be repeated through the join: for example, one patient may have multiple admissions. In Figure~\ref{fig:example}, $Coverage({\small\texttt{PATIENT}}\bowtie{\small\texttt{ADMISSION}})=%
 \frac{1}{2}\cdot\big(\frac{7}{5}+\frac{4}{7}\big)\approx 0.99$.

\subsubsection{Completeness and correctness of \jedi with Selective Mining}

Algorithm 1 either checks if the tuples violating the AFDs are excluded by the join operation or, if AFDs are not provided, it mines the exact FDs after the tuples are filtered by the join operation. Algorithm 2 retrieves the minimal FDs from the logically inferred FDs, thus no \texttt{lhs} subset of FDs remains unchecked. Theorem~\ref{theo6} indicates which part of the candidate FD lattice can be pruned. Then, Algorithm 3 explores the candidate FDs using a classic bottom-up approach, thus no minimal FD remains unchecked. Overall, our algorithms explore the lattice of candidate FDs until they find minimal FDs; they avoid only the parts of the lattice that do not contain valid candidate FDs, thus \jedi with selective mining retrieves the complete set of minimal candidate FDs.

\begin{figure}[h]
	\centering
	\includegraphics[width=\linewidth]{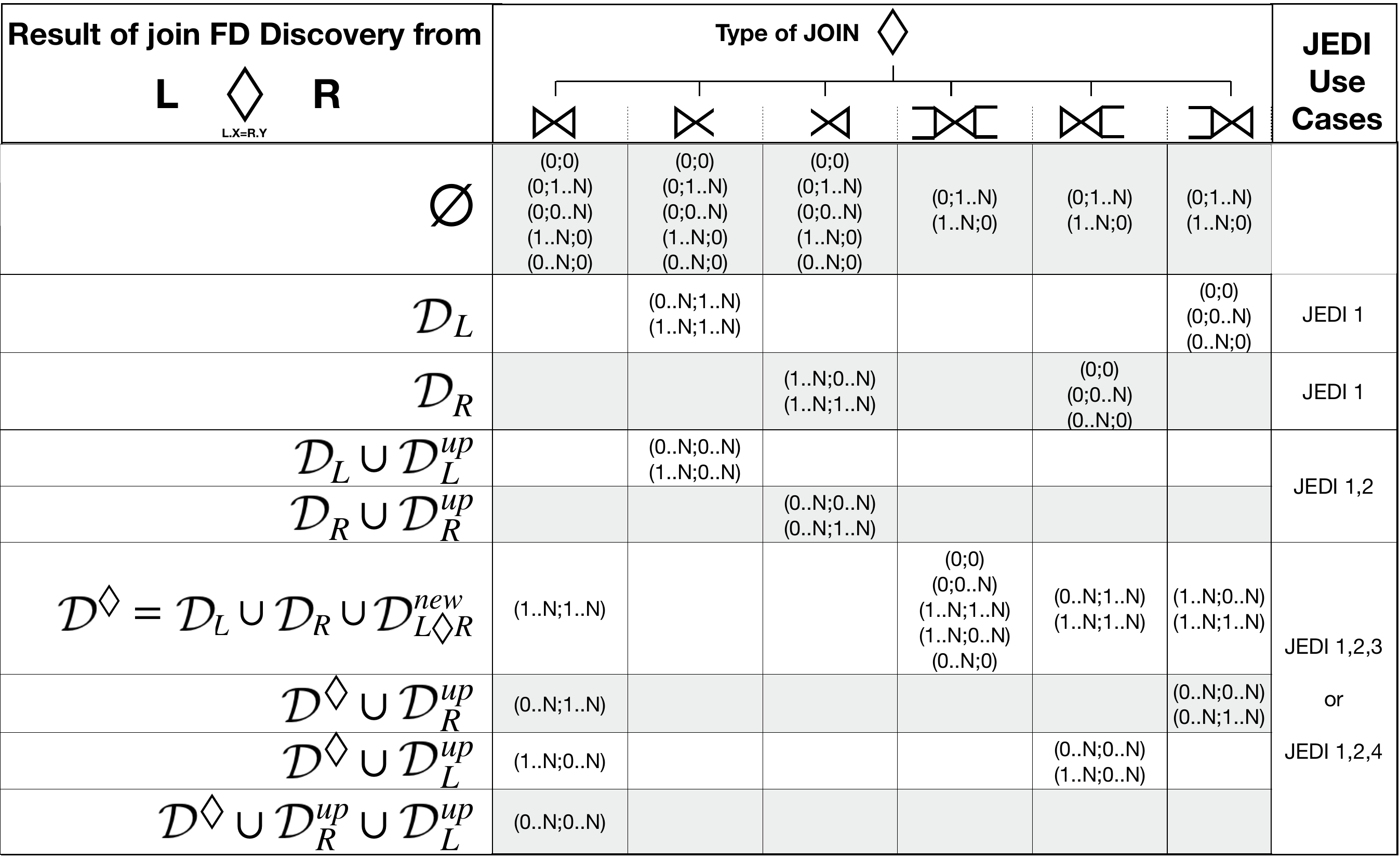}
	\caption{Computation of join FD sets depending on the join type and cardinalities of join attributes in $L\Diamond R$: (0;1..N) means that some tuples are absent in $L$, whereas they are present once or multiple times in $R$; (0..N;1..N) means that some tuples are absent or repeated multiple times in $L$, whereas they are always present in $R$ at least once and multiple times.}\label{fig:cardinalities}
\end{figure} 
In Algorithm 1, valid exact FDs are either discovered from the data or deduced (because the join operation filters the tuples violating some FDs that became consequently exact after the join). Both cases of FDs are guaranteed to hold on the joined instance and only minimal \texttt{lhs} are kept. Theorem~\ref{thm:MonoInstanceCrossJoinFDs} shows the correctness of the set of FDs inferred through logical inference. Then, the subroutine \subRefine checks the correctness of its candidates FDs holding on the data. Therefore, the set of FDs $\fdSet_{2}$ discovered by Algorithm 2 is such that  $\fdSet_{2} \models \fdSet_{\Diamond}$.  In Algorithm 3, Theorem~\ref{theo6} enforces the retrieval of FDs based only on the attributes that can become \texttt{rhs} in the joined instance, then only plausible candidate FDs are explored. Therefore, every discovered FD holds in the joined instance. By construction, Algorithm 2 and 3 lead to the retrieval of FDs with minimal \texttt{lhs} only.

\subsubsection{Completeness and Correctness of \jedi with sampling}
\jedi with sampling uses Algorithm 1 and 2 with the same completeness and correctness guarantees exposed previously for \jedi with selective mining.
 Algorithm~4  uses Algorithm 3 over samples, the retrieved FDs are minimal exact FDs holding on the samples. If there exists a minimal FDs $A\to b$ in a sample, an FD $A^\prime \to b$ such that $A^\prime \subset A$ cannot become a minimal FD in the full instance. Therefore, the FDs retrieves in from the samples imply the set of FDs over the full joined instance.
The correctness of Algorithm~\ref{algo:statisticalInferenceJoinFDs}  cannot be guaranteed as there is a probability of not sampling a pair of counter-example tuples for a given AFD, and thus to confound the AFD with an exact FDs in the join results. Due to the NP-completeness of the minimal FD mining problem\cite{davies1994np}, every pair of tuples must be checked in order to find a possible counter-example, thus to guarantee that the discovered FDs from samples are exact FDs. We would like to note that developing a framework to guarantee the accuracy of the sampling-based FD discovery approach is a  challenging problem and a focus of our future research.

\subsubsection{Complexity}  Algorithm 1 subroutine \texttt{updstagedFDs} is based on a level-wise algorithm through the attributes lattice. Its complexity is exponential in the number of attributes of the considered table. It prunes candidates at each level when it is possible. In terms of memory, only two levels are required. The memory size is bounded by $\mathcal{O}\binom{k}{k/2}$ where $k$ is the number of attributes. The complexity of the second subroutine \texttt{upstagedAFDs} of Algorithm 1 is $\mathcal{O}(n \cdot f)$ where $n$ is the maximal number of tuples and $f$ the maximal number of FDs either from the left or the right table.  The join computation is linear because we use a merge join algorithm over indexed data. It should be noted that subroutine \texttt{upstagedAFDs} does not compute the full join but instead a subset of the full join (Algorithm~\ref{algo:MonoInst_CrossJoinFDs} line\#~\ref{line:instRef}). %
Algorithm 2 infers and refines FDs coming from the previous step with complexity $\mathcal{O}(n \cdot f)$, where $f$ is the number of validated FDs and $n$ the maximal number of tuples in the left or right instance. 
The complexity of Algorithm 3 is $\mathcal{O}(f \cdot f_j)$ where $f$ is the maximal number of validated FDs in the left or right instance and $f_j$ the number of validated FDs in the join instance. %
Algorithm 4 operates on micro-joins computed by {\small\texttt{micro\_join}} subroutine with complexity $|P_i| * ( O(|L_i| \cdot |R_i| ) +  O( |\mathcal{E}_{\mathcal{D}}|)$ with $|P_i|$ the number of tuples in the micro-join $i$, $|L_i|$ and $|R_i|$ the number of tuples in $L$ and $R$ participating to the micro-join, and  
$\mathcal{E}_{\mathcal{D}}$,  the set of consistent FDs discovered from multiple micro-joins FD sets. The complexity of selective sampling is $ O( |L| + |R| )  + O( n_v * k_{max} )$, with $k_{max}$, the maximum number of attributes in $L$ or $R$.

\section{Experiments}

\label{sec:experiments}

{\bf Evaluation Goal.} We compare the two variants of our method which computes only necessary micro-joins on-the-fly against the straightforward approach that consists of computing first the full join of two or more tables and then mine all the FDs from the joined result.  The two main points we seek to validate are: (1) Does our approach enable us to discover join FDs accurately (i.e,. high precision) in an efficient manner and faster than the straightforward approach? (2) What is the impact of different data and joins characteristics on \jedi performance? %

{\bf Setup.} We perform all experiments on a laptop Dell XPS machine with an Intel Core i7-7500U quad-core, 2.8 GHz, 16 GB RAM, powered by Windows 10 64-bit. %
Our algorithm implementation in Java use only one thread. Our code, scripts, and data sets are available at \url{https://github.com/xxxx}. 

{\bf Methods.} We compare \jedi algorithms for selective mining (\jedi\_SM) and sampling-based join FD discovery (\jedi\_SB) against with four state-of-the-art FD discovery methods: (1) TANE \cite{HKPT98,HKPT99},   (2) Fast\_FDs \cite{WGR01}, and (3) FUN \cite{NoCi01icdt}, and (4) HyFD \cite{HyFD}, using Java implementation of Metanome \cite{Papenbrock:2015}. %
Data sets are stored in a Postgres DBMS. Join attributes are indexed with B-Tree and hash indexes. %

\begin{table}[t]
\centering
\scriptsize
\begin{tabular}{clrr}
\hline
{\bf Data set} & {\bf Table}& {\bf (Att\# ; Tup\#) }&{\bf FD\#}\\
\hline
\hline
 &{\bf Patients}& (7 ; 46.52k)& 11 \\
MIMIC-III& {\bf Admissions}& (18 ; 58.976k)& 285\\
&{\bf Diagnoses\_icd}& (4 ; 651.047k)  & 2\\
&{\bf D\_icd\_Diagnoses}& (3 ; 14.710k)& 2\\
\hline
&{\bf pte\_active}& (2; 300)& 1\\
PTE&{\bf pte\_bond}&(4 ; 9.317k)& 3\\
&{\bf pte\_atm}& (5 ; 9.189k)& 5\\
&{\bf pte\_drug}& (1 ; 340)& 0\\
\hline
&{\bf atom}& (3; 12.333k)& 2\\
PTC&{\bf connected}&(3 ; 24.758k)& 3\\
&{\bf bond}& (3 ; 12.379k)& 2\\
&{\bf molecule}& (2 ; 343)& 1\\
\hline
&{\bf Supplier}& (7 ; 10k)& 34 \\
TPC-H&{\bf Customer}&(8 ; 150k)& 51\\
&{\bf Nation}& (4 ; 23)& 9 \\
&{\bf Region}& (3 ; 5)& 6 \\
&{\bf Part}& (7 ; 200k)& 99 \\
&{\bf Partsupp}& (5 ; 800k)&11 \\
\end{tabular}
\caption{Data characteristics for our experiments.
}\label{exp:tables}
\end{table}
\begin{figure}[t]
	\centering
	\includegraphics[height=6cm]{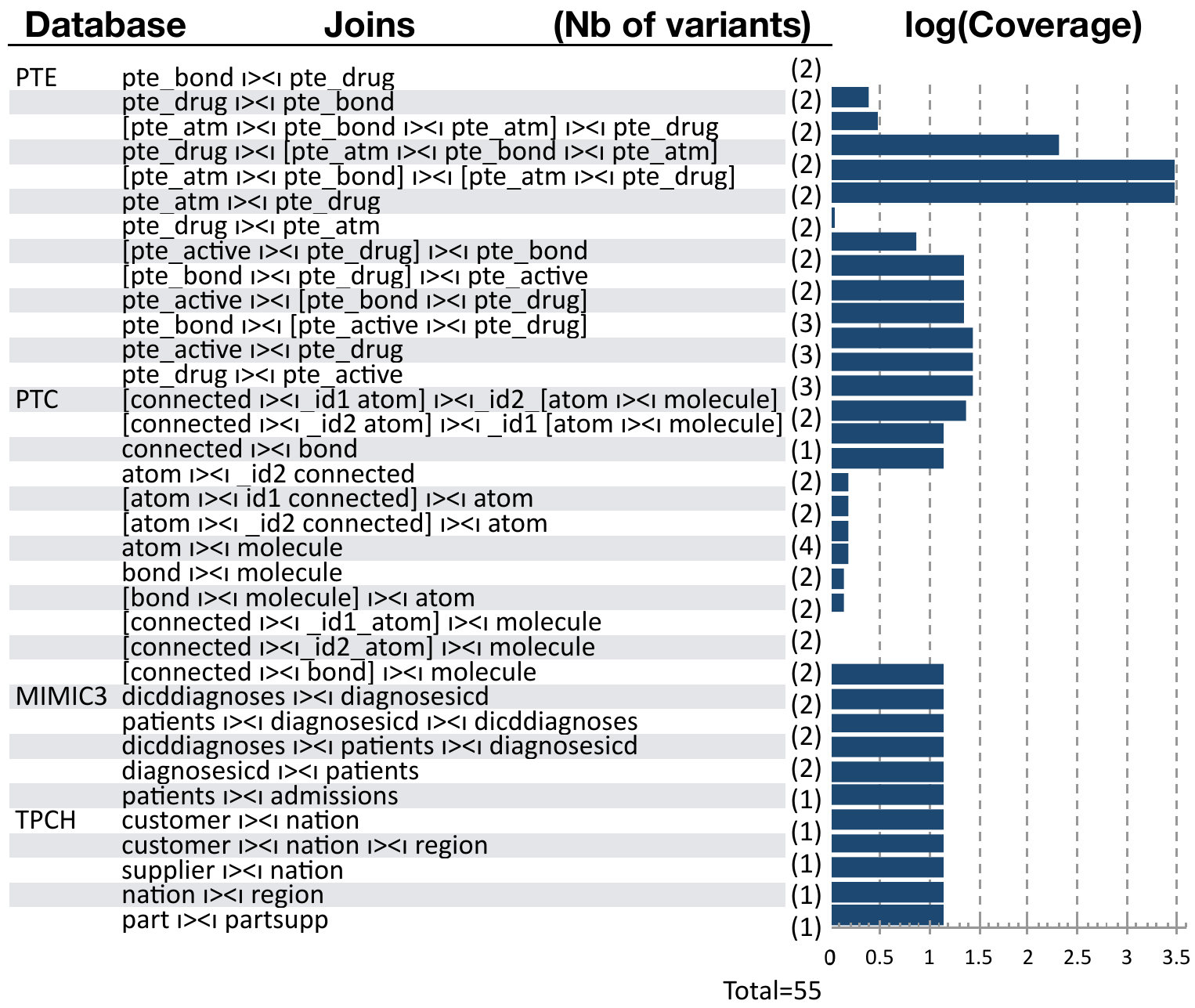}
	\caption{Coverage of the join operations considered in our experiments. The number of variants of each join operation (with different orderings) is indicated in parentheses.}\label{fig-cov-join}
\end{figure} 

{\bf Evaluation Metrics.}
To measure the performances of join FD discovery, we use precision defined as the fraction of correctly discovered FDs on-the-fly  by \jedi from partial joins that are common with the true FDs discovered from the full join result by the straightforward method over the total number of discovered true FDs. %
For each experiment, we consider ten runs per setting and report the average performance of each method. We examine precision, runtime, and memory consumption for different join operators with varying: (1) the size of the input tables and joins, (2) the cardinalities of the join attribute values and join coverage, and (3) the sample size.

{\bf Data.}
We use  three  real-world  datasets  and  one  synthetic data set in our experiments: (1) the clinical database {\tt MIMIC-3}\footnote{\url{https://physionet.org/content/mimiciii/1.4/}} \cite{mimiciii}; 
 (2) PTE\footnote{\url{https://relational.fit.cvut.cz/dataset/PTE}}, a database for predictive toxicology evaluation, used to predict whether the compound is carcinogenic, and (3) PTC\footnote{\url{https://relational.fit.cvut.cz/dataset/PTC}}, the data set from the Predictive Toxicology Challenge that consists of more than three hundreds of organic molecules marked according to their carcinogenicity on male and female mice and rats; and (4) the TPC-H Benchmark\footnote{\url{http://www.tpc.org/tpch/}} with scale-factor 1. 
Data sets characteristics are given in Table \ref{exp:tables} and the characteristics of the join operations in Figure 4.

\begin{figure*}[!t]
	\centering
	\includegraphics[width=\linewidth]{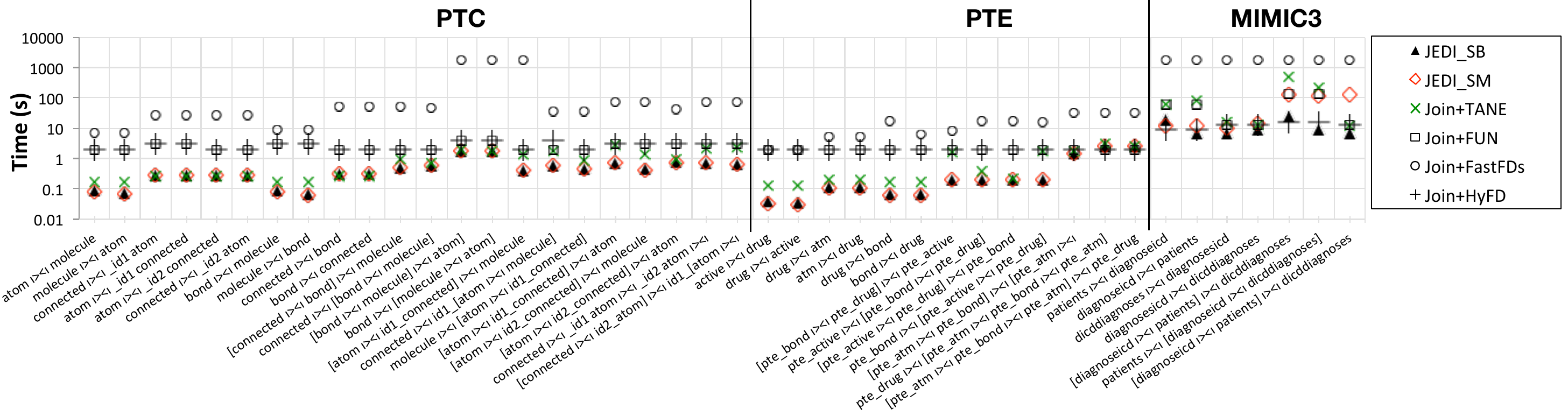}
	\caption{Average runtime of \jedi against FUN, TANE, FastFDs, and HyFD algorithms}\label{fig:jedi-time-comparison}
\end{figure*} 
 \begin{figure*}[t]
	\centering
	\includegraphics[width=\linewidth]{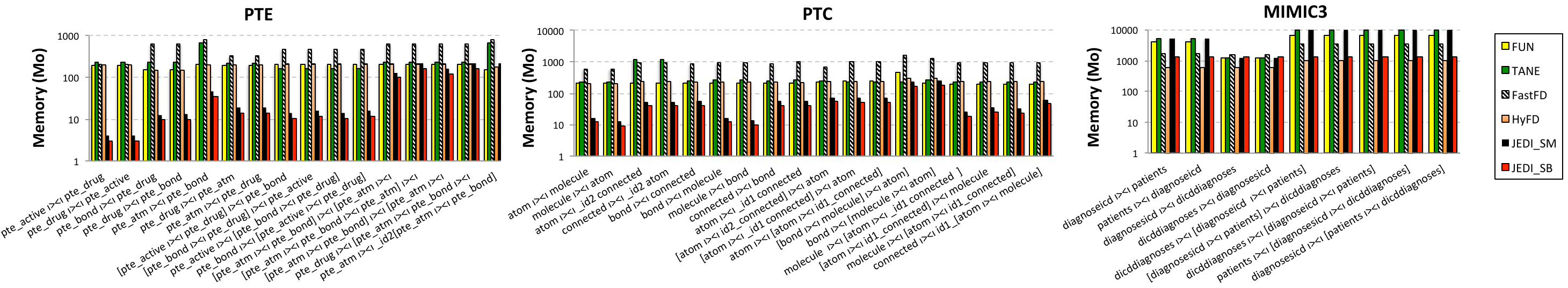}
	\caption{Maximal memory consumption of \jedi against FUN, TANE, HyFD and FastFDs  algorithms}\label{fig:jedi-memory}
\end{figure*} 

\subsection{Efficiency and Effectiveness Evaluation}
In a first set of experiments, we evaluate the runtime and memory consumption of \jedi algorithms compared to the state-of-the-art FD discovery methods that follow the straightforward approach over a total of 55 join operations on the real-world and synthetic data sets with various coverage values illustrated in Figure~\ref{fig-cov-join}. Variants of the same join operation with different orderings of the participating tables have the same coverage.  We want to evaluate the gain of the key components of \jedi for join FD discovery and therefore focus on \jedi Algorithm 1 (join FD discovery from upstaged AFDs), Algorithm 2 (logical inference), and  Algorithm 3 (selective mining of the remaining FDs from the join result) for \jedi\_SM  and  algorihms 1,2, and 4 (selective sampling) for \jedi\_SB. %

\subsubsection{Runtime} Figure~\ref{fig:jedi-time-comparison} shows the average runtime (in seconds) for the two variants of \jedi. %
We log the average total runtime of join FD discovery (including data loading) for all methods over 10 runs. For the competing methods, we  added the average execution time of each join operation over the indexed data.  For PTE and PTC data sets with low coverage, \jedi\_SM is much faster than the other methods operating on the pre-computed full join results with one order of magnitude on average. %
However, for MIMIC3 and joins with higher coverage, \jedi\_SM is not as competitive as HyFD applied to the full join result. This result actually motivated the design of a sampling-based method that can be more efficient when join coverage is high ($\log(Coverage)>1$) in large table sizes (i.e., when many tuples from one table are repeated in the join operation).  \jedi\_SB shows similar or better performances for low coverage joins PTE and PTC with a sampling size around 28\% in average to reach precision of 1. However, when both join coverage and table size increase, \jedi\_SB  needs 63\% of the sample size on average to reach precision 1 as shown for MIMIC3 to compete with HyFD that operates over the precomputed full join. Also, the figure shows that both variants of \jedi are sensitive to the table ordering in the join operation, even more when the join size increases. %

\subsubsection{Memory Consumption} As shown in  Figure~\ref{fig:jedi-memory}, average maximal memory consumption of \jedi\_SM is the lowest for PTE and PTC data sets with low coverage. Memory consumption of \jedi\_SB (when precision is 1) is between 23.4\% and 34.7\% lower than \jedi\_SM on average for PTE and PTC and, which makes it perform better than HyFD in certain cases. However, for MIMIC3 as the number of attributes and the coverage increase, HyFD generally outperforms the other methods with \jedi\_SB has the second position.  %

 \begin{figure*}[t]
	\centering
		\includegraphics[width=.98\linewidth]{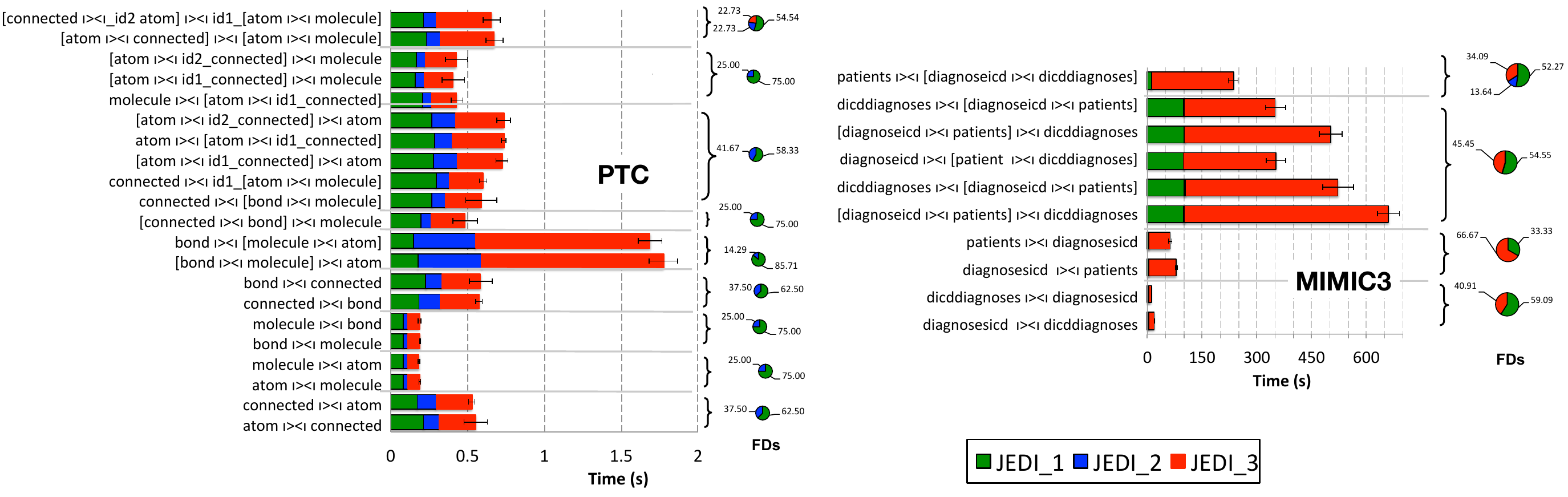}
		
	\caption{Average runtime of \jedi\_SM with breakdown per algorithm}\label{fig:jedi-time-repartition}
\end{figure*} 

\begin{figure*}[ht]
	\centering
	\includegraphics[width=.99\linewidth]{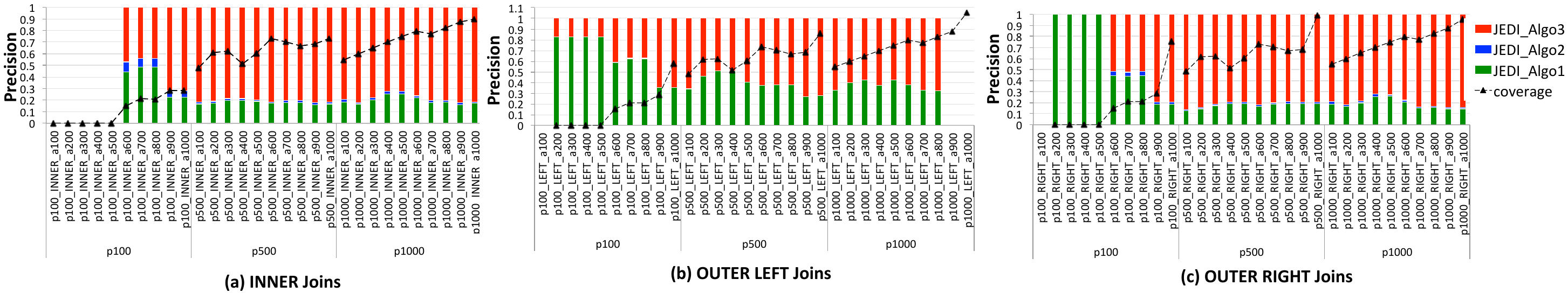}
	\caption{Effect of coverage on the precision of \jedi\_SM algorithms 1,2, and 3 for (a) inner, (b) outer left, and (c) outer right joins between \texttt{PATIENT} and \texttt{ADMISSION} with [100,500,1K] and [100, 200,...,1K] tuples respectively.}\label{fig:inner}
\end{figure*} 
\subsection{Evaluation of \jedi with Selective Mining }
In Figure~\ref{fig:jedi-time-repartition},  we report the average runtime breakdown of each algorithm 1, 2, and 3 in \jedi\_SM (in the horizontal histograms) and their respective percentages of discovered FDs (in the adjacent pie charts with the same color coding). Error bars represent the standard deviation of the average total runtime of the algorithms over ten runs.  The data distributions, cardinalities of the join attribute values, and the join coverage have very different characteristics across the data sets to illustrate the behavior of our algorithms on different join orderings and different coverage values for PTC and MIMIC3. We observe the same behavior in PTE which is not shown due to space limitation. %
Algorithm 1 manipulates different sets of upstaged AFDs per input table, as we can see with the 3-size joins in MIMIC3. The logical inference made by Algorithm 2 from Algorithm 1 output will lead to different resulting sets of inferred FDs. For MIMIC3, \jedi 1 alone retrieves 52\% of the true FDs on average. Algorithm 2 helps only for {\small\texttt{patients}} $\bowtie$ [{\small\texttt{diagnoseicd}} $\bowtie$ {\small\texttt{dicddiagnoses}}], the only case when the logical inference can be triggered from the FDs discovered by Algorithm 1.  

 The figure shows that different orderings in the tables participating in the join operations change not only the runtime but also the contribution of each algorithm in retrieving the true FDs. The reason is that various orderings of the input tables lead to different sets of potential upstaged AFDs and  trigger different logical inferences. Consequently, the remaining set of join FDs to discover (by Algorithm 3 or 4) will be different. Future work will be to find the optimal ordering of the input tables %
 and predict the cases where logical inference could prevail over selective mining to discover FDs more efficiently. %

\subsection{Coverage Analysis}
In this experiment reported in Figure \ref{fig:inner}, we use subsets of MIMIC3 database with 100, 500, and 1000 tuples of {\small\texttt{PATIENT}} joined with 100 to 1000 tuples of  {\small\texttt{ADMISSION}} with (a) inner, (b) outer left, and (c) outer right join operators for {\small\texttt{PATIENT}}$\Diamond${\small\texttt{ADMISSION}}. We evaluate the precision of the three algorithms of \jedi\_SM (Y-axis) as the coverage rate (X-axis) increases across 30 joins for each join operator (90 joins). Coverage increases with the size of the \texttt{rhs} table participating in the join operation which indicates cardinalities $(0..N;0..N)$ as more and more tuples are repeated through the join operation.  Precision of Algorithm  1 decreases as the coverage and sizes of both \texttt{lhs} and \texttt{rhs} tables of the join increase for the three types of join showing the difficulty of leveraging AFDs when the number of repetitions increases. 
 \jedi 1 precision is relatively low for inner and outer right joins with best values (around .45) when coverage is below .5. Precision of  Algorithm  2 is close to 0 because logical inference cannot be leveraged from Algorithm 1's outputs, when the coverage and sizes of the joined table increase. We observed that Algorithms 1 and 2 tend to be more effective when the coverage is relatively low as shown in Figure \ref{fig:inner} for the three types of join operations.

\subsection{Evaluation of \jedi with Sampling }
In this last set of experiments, we evaluate \jedi 1,2,4 and show that it outperforms \jedi\_SM and other competing methods in terms of runtime for precision equals to 1, in particular for the worst case scenarios when coverage is high. First, we investigate the effect of the sample size on the precision of \jedi Algorithms 1,2, and 4. 
\begin{wrapfigure}{l}{0.25\textwidth}
	\centering
	\includegraphics[width=\linewidth]{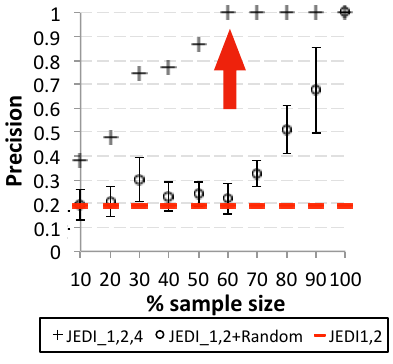}
	\caption{Effect of sampling size on \jedi\_SB precision for  {\small\texttt{PATIENT}\_{{\small 1K}}} $\bowtie$ {\small\texttt{ADMISSION}\_{{\small 1K}}}  }\label{fig:sample-size}
\end{wrapfigure}
\noindent Figure \ref{fig:sample-size} shows precision of selective sampling with \jedi 1,2,4 compared to \jedi 1,2 followed by random sampling (averaged over 10 runs) for the inner join between 1,000 tuples of  {\small\texttt{PATIENT}}  and 1,000 tuples of {\small\texttt{ADMISSION}} with high \mbox{coverage (>1)} (cf. Figure \ref{fig:inner}(a)).  For this experiment, we used $n_b=1$ as input parameter for selective sampling to pick only one tuple as the representative tuple of a branch and we vary $n_v$ from 0 to $| atts({\small\texttt{PATIENT}}\bowtie{\small\texttt{ADMISSION}})|$ to exclude the attributes with the highest number of  distinct values. 
\jedi 1,2,4 reaches precision (and recall) of 1 when the size of the selected sample is 62.3\% of the table input size. Interestingly, even for a sample of size of 10\% with tuples selected by Algorithm 4, we can see an improvement of +0.153 over \jedi 1 and 2 precision.

\begin{table*}[t]
\scriptsize
\centering
\begin{tabular}{|c|l|p{.8cm}|r|r|r||r|r|r|p{1.3cm}|}
\cline{5-10}
\multicolumn{4}{c|}{}&\multicolumn{6}{c|}{{\bf Average Runtime (seconds)}}\\
\hline
{\bf Database}&{\bf Join}&{\bf Coverage}&{\bf Join Size Ratio}&{\bf JEDI\_SM}&{\bf JEDI\_SB}&{\bf JOIN+HyFD} &{\bf JOIN+FUN}&{\bf JOIN+TANE}&{\bf JOIN+FastFDs}\\
\hline

MIMIC3&diagnoseicd $\bowtie$ patients&\multirow{2}{*}{7.50}&429,530/651,047 (66\%) &79.86$\pm$2.33& 12.73 $\pm$ 0.85&\multirow{2}{*}{{\bf 10.23$\pm$0.53}}&\multirow{2}{*}{61.23$\pm$3.35}&80.14$\pm$4.07&\multirow{7}{*}{>1800}\\
\cline{5-6} \cline{9-9} 

&patients $\bowtie$ diagnoseicd&&422,120/651,047 (65\%)&62.35$\pm$3.16& {\bf 10.12 $\pm$ 0.34}&&&62.78$\pm$3.26&\\
\cline{2-9} 

&dicddiagnoses $\bowtie$ diagnosesicd&\multirow{4}{*}{22.84}&430,355/634,709	(68\%)&10.72$\pm$0.42&{\bf 7.30} $\pm$ 0.94&\multirow{2}{*}{10.32$\pm$0.51}&\multirow{2}{*}{13.21$\pm$0.15}&17.98$\pm$2.34&\\
\cline{4-4}\cline{5-6} \cline{9-9}
&diagnosesicd $\bowtie$ dicddiagnoses&&400,314/634,709	(63\%)	&19.25$\pm$0.66&{\bf 10,17} $\pm$ 1.58&&&18.41$\pm$ 0.43&\\
\cline{4-4} \cline{5-6} \cline{7-9}
&[diagnoseicd $\bowtie$ patients] $\bowtie$ dicddiagnoses&&410,452/634,709	(65\%)&472.57$\pm$31.64&{\bf 15.64 $\pm$	1.73}&\multirow{2}{*}{\bf 16.32$\pm$1.33}&\multirow{2}{*}{ 140.07$\pm$12.15}&473.85$\pm$27.23&\\
\cline{4-4}\cline{5-6} \cline{9-9} 
&patients $\bowtie$ [diagnoseicd $\bowtie$ dicddiagnoses]&&424,273/634,709	(67\%)&220.67$\pm$12.88&17.55 $\pm$	1.18&	&&221.44$\pm$10.98&\\
\cline{4-7}
\hline\hline
TPCH&part $\bowtie$ partsupp&2.5&525,752/800K (22\%)&721.08$\pm$47.41 &{\bf 3.37 $\pm$ 0.52}&3.52$\pm$0.54 &149.07$\pm$15.25&\multicolumn{2}{c|}{>2100}\\
\cline{2-10}
&nation  $\bowtie$ region&3&17/25	(65\%)& 14.36$\pm$0.41&{\bf 3.07 $\pm$ 0.17}& 3.32$\pm$0.51  &3.03$\pm$0.01&{\bf 3.07$\pm$0.12}&3.39$\pm$0.17\\
\cline{2-10}
&supplier $\bowtie$ nation&200.5&6,524/10K	(66\%)	&	3.06$\pm$0.01&   {\bf 3.02 $\pm$ 0.28}&3.62$\pm$0.57  &3.06$\pm$0.21&3.33$\pm$0.27&3.43$\pm$0.27\\
\cline{2-10}
&customer $\bowtie$  nation&\multirow{2}{*}{3000.5}&101,076/150K	(45\%)&10.79$\pm$0.18&3.49	$\pm$ 0.13& {\bf 3.38$\pm$0.52} & 3.50$\pm$0.65&3.44$\pm$0.33&3.64$\pm$0.23\\
\cline{2-2} \cline{4-10}
&customer $\bowtie$ nation  $\bowtie$ region&&102,162/150K (77\%)&13.59	$\pm$ 0.69&{\bf  3.55$\pm$0.25}& 45.52$\pm$1.44  & 3.56$\pm$0.85&3.72$\pm$0.24&3.75$\pm$0.14\\
\hline
\end{tabular}
\caption{Average runtime comparison on MIMIC3 and TPC-H data.}\label{tab:exp-jedi4}
\end{table*}

This shows the main advantage of our approach, reducing drastically the execution time of the join FD discovery.
Table \ref{tab:exp-jedi4}  completes the comparative study of the two variants of \jedi against the competing methods for the worst case scenarios when join coverage is high for large table size. Recall that \jedi\_SM has shown the best performances over the state-of-the-art methods for PTE and PTC joins with relatively low join coverage and small table sizes (from 2 to 5 attributes and 300 to 24k tuples). So, we focus on MIMIC3 and TPC-H joins and report for each join: its coverage, the Join size ratio (expressed as  the sample size over the full join size, and as a percentage in parentheses, and the average runtime of \jedi\_SM, \jedi\_SB  at precision 1, FUN, TANE, HyFD, and FastFDs. 

For MIMIC3, we observe that \jedi\_SB outperforms all the methods with at least one of the variants of the join ordering, otherwise HyFD applied to the full join result is slightly faster. This shows that the ordering of large tables is critical for size 3 and more joins but most importantly, high coverage impacts the sampling strategy. The reason is that many tuples from each input table are repeated through the join operation and the choice of more than one representative tuples jeopardizes the sampling result, in particular when it does not preserve the balance and distributions of the repeated tuples from the full join to the multiple micro-joins. 

For TPC-H, we observe the same phenomenon for the 3-size join  where the straightforward approach with HyFD is slighlty more efficient (or equal with TANE). We observe that the join size ratio to reach precision 1 is influenced by the coverage and also the number of distinct values for each attributes of the joined tables. For example, for the join {\small \texttt{part}}$\bowtie {\small \texttt{partsupp}}$, the average number of distinct values per attribute is 52,984 for {\small \texttt{part}}  and 223,971  for {\small \texttt{partsub}} which facilitates the sampling strategy in selecting violating tuples and reducing the space of candidates FDs. This can explain the relatively low ratio (22\%) despite high join  coverage.

These observations show that there is room for improvement  in the two following directions: (1) to find the optimal ordering of the tables to maximize the performance of FD discovery (in particular in Algorithm 1) and and (2) to improve the sampling strategy and better tune the sample size $n_b$ to maximize the number of violations and reduce the space of candidate join FDs.
\section{Related Work} 
\label{sec:relatedwork}
In the last three decades, numerous approaches from the database and the data mining communities have been proposed to extract automatically valid exact  and approximate FDs from single relational tables~\cite{KiMa95,caruccio2015relaxed}. Liu et al. \cite{Liu2012} have shown that the complexity of FD discovery is in $O(n^2(\frac{k}{2})^22^k)$ where $k$ is the number of attributes and $n$ the number of records considered.   %
 To find FDs efficiently, existing approaches can be classified into three categories:~(1) Tuple-oriented methods (e.g., FastFDs~\cite{WGR01}, DepMiner~\cite{lopes2000efficient}) that exploit the notion of tuples agreeing on the same values to determine the combinations of attributes of an FD; (2) Attribute-oriented methods (e.g., Tane~\cite{HKPT98, HKPT99}, Fun~\cite{NoCi01icdt, NoCi01is}, FDMine~\cite{YH08}) that  use pruning techniques and reduce the search space to the necessary set of attributes of the relation to discover exact and approximate FDs.    HyFD~\cite{PapenbrockN16} exploits simultaneously the tuple- and attribute-oriented approaches to outperform the previous approaches; and more recently 
(3) Structure learning methods relying on sparse regression \cite{Zhang2020}, or on entropy-based measures \cite{kenig2019mining} to score candidate constraints (not limited to FDs alone). More particularly, FDX \cite{Zhang2020}  performs structure learning over a sample constructed by taking the value differences over sampled pairs of tuples from the raw data. %
In addition, incremental approaches~\cite{SchirmerP0NHMN19, CaruccioCDP19} have been developed to tackle data volume and velocity with updating all valid FDs when new tuples are inserted outperforming classical approaches that recalculate all FDs after each data update. Extensive evaluations of FD discovery algorithms can be found in \cite{DurschSWFFSBHJP19,PapenbrockEMNRZ15}. To the best of our knowledge, previous work on FD discovery did not attempt to address the problem join FD discovery in an efficient manner. Our approach combining logical inference, partitioning,  and selective sampling is the first solution in this direction.

\section{Conclusions}
\label{sec:conclusion}
We introduced {\textsf{JEDI}}, a framework to solve the problem of FD discovery from multiple joined tables without the full computation of the join result beforehand. The salient features of our work are the following: (1) We leverage single-table approximate FDs that become exact join FDs due to the join operation when the join value sets are not preserved; (2) We leverage logical inference to discover join FDs from the sets of single-table FDs without computing the full join result; and (3) We find new multi-table join FDs from partial join and micro-joins using respectively selective mining on the necessary attributes and selective sampling on the necessary tuples. We empirically show that \jedi outperforms,  both in terms of runtime and memory consumption, the state-of-the-art FD discovery methods applied to the join results that have to be computed beforehand.
We hope that this will open a new line of research where the community will examine FDs across multiple tables and an entire database. Since join FDs  are resilient to the join operators, they can be leveraged for pruning  and assessing the validity and prevalence of a large set of discovered single-table FDs.

\newpage

\bibliographystyle{abbrv}
\bibliography{biblio.bib} 

\end{document}